\documentclass[]{interact}

\usepackage{epstopdf}
\usepackage[caption=false]{subfig}

\usepackage[numbers,sort&compress]{natbib}
\bibpunct[, ]{[}{]}{,}{n}{,}{,}
\renewcommand\bibfont{\fontsize{10}{12}\selectfont}
\makeatletter
\def\NAT@def@citea{\def\@citea{\NAT@separator}}
\makeatother

\theoremstyle{plain}
\newtheorem{theorem}{Theorem}[section]
\newtheorem{lemma}[theorem]{Lemma}

\theoremstyle{definition}

\theoremstyle{remark}

\usepackage{amsmath,amssymb,amsthm}
\usepackage[linesnumbered,vlined,ruled]{algorithm2e}
\usepackage{array}

\usepackage{fancyhdr} 
\fancyheadoffset{0cm}

\fancypagestyle{plain}{%
  \fancyhf{}%
  \fancyhead[R]{\thepage}%
}

\begin{document}


\title{Arbitrary pattern formation by opaque fat robots on infinite grid}

\author{
\name{Manash Kumar Kundu \textsuperscript{a}\thanks{CONTACT Manash Kumar Kundu Email: manashkrkundu.rs@jadavpuruniversity.in}, Pritam Goswami \textsuperscript{b}\thanks{CONTACT Pritam Goswami Email: pritamgoswami.math.rs@jadavpuruniversity.in}, Satakshi Ghosh\textsuperscript{b}\thanks{CONTACT Satakshi Ghosh Email: satakshighosh.math.rs@jadavpuruniversity.in} and Buddhadeb Sau \textsuperscript{b}\thanks{CONTACT Buddhadeb Sau Email: buddhadeb.sau@jadavpuruniversity.in}}
\affil{\textsuperscript{a}Gayeshpur Government Polytechnic, Department of Science and Humanities, Kalyani, West Bengal - 741234, India; \textsuperscript{b}Jadavpur University, Department of Mathematics, Kolkata , West Bengal - 700032, India.}
}

\maketitle

\begin{abstract}
Arbitrary Pattern formation ($\mathcal{APF}$) by a swarm of mobile robots is a widely studied problem in the literature. Many works regarding $\mathcal{APF}$ have been proposed on plane and infinite grid by point robots. But in practical application, it is impossible to design point robots. In \cite{BoseAKS20}, the robots are assumed opaque fat robots but the environment is plane. To the best of our knowledge, no work till now ever considered the $\mathcal{APF}$ problem assuming opaque fat robots on infinite grid where movements are restricted. In this paper, we have provided a collisionless distributed algorithm and solved $\mathcal{APF}$ using 9 colors.


\end{abstract}


\section{Introduction}
Nowadays, the distributed system is gaining popularity among researchers due to its many positive aspects. Designing a  centralized system and also maintaining its robustness is not at all cost-effective. But these factors can be handled easily and effectively in a distributed system. Swarm robotics is an example of such distributed system. In swarm robotics, more than one robot is considered in an environment (plane, network etc.). The robots are considered to be autonomous ( i.e they do not have any central control), homogeneous (i.e all robots execute the same algorithm) and identical (they are not indistinguishable by their physical appearance). Before the study of swarm robotics, designing a robot to do a specific task was costly as it would have needed many strong capabilities. But designing a swarm of robots is cheaper than using such robot with many capabilities as the goal now become to design the robots with minimal capabilities such that they can do the same task autonomously. Among many applications of swarm robots military operations, border surveillance, cleaning of a large surface, rescue operations, disaster management etc. are the ones that use the swarm robots vividly in present days. So, it is evident why swarm robotics has gained such popularity in the industry and among researchers in the current scenario. 

Among many problems (eg. gathering, scattering, exploration) Arbitrary Pattern Formation ($\mathcal{APF}$) is a classical problem in the field of swarm robotics. In this problem, a swarm of robots which are deployed in an environment (plane, graph etc.), need to form an already decided pattern which is given as input to the robots. The robots can move freely in the plane but in the case of a graph, they always move through an edge of the graph. There are mainly four models of robots depending upon their capabilities. These models are $\mathcal{OBLOT}$, $\mathcal{FSTA}$, $\mathcal{FCOM}$ and $\mathcal{LUMI}$. In all of these models, robots are considered to be autonomous, homogeneous, identical and anonymous (i.e the robots do not have any unique identifiers). In the $\mathcal{OBLOT}$ model, the robots are considered to be oblivious (i.e the robots do not have any persistent memory to remember any previous state) and silent (i.e the robots do not have any means of communication among themselves). In the $\mathcal{FSTA}$ model, the robots are silent but not oblivious. In the $\mathcal{FCOM}$ model, the robots are oblivious but not silent. And in the $\mathcal{LUMI}$ model, the robots are neither silent nor oblivious. There are many works of $\mathcal{APF}$ which has considered the $\mathcal{OBLOT}$ model in literature \cite{BoseAKS20,BoseKAS21,BoseAKS20grid,FlocchiniPSW08,AdhikaryKS21}. 

Activation time of the robots plays an important role to design algorithms for a swarm of robots. It is assumed that a scheduler controls the activation of robots during the execution of any algorithm. Mainly, there are three types of schedulers that have been considered in many previous works. These schedulers are $\mathcal{FSYNC}$ or a fully synchronous scheduler, $\mathcal{SSYNC}$ or a semi-synchronous scheduler and $\mathcal{ASYNC}$ or an asynchronous scheduler. In the case of a $\mathcal{FSYNC}$ scheduler, time is divided into global rounds of the same duration and each robot is activated at the beginning of each round. $\mathcal{SSYNC}$ scheduler is a more general version of the $\mathcal{FSYNC}$ scheduler. In the case of the $\mathcal{SSYNC}$ scheduler, time is divided into global rounds of the same duration as it has been done for $\mathcal{FSYNC}$ scheduler. But at the beginning of each round, the set of activated robots can be a proper subset of the set of all robots (i.e. all robots may not get activated at the beginning of each round). Now, in the case of the $\mathcal{ASYNC}$ scheduler, there is no sense of global rounds. Any robot can get activated at any time. So, the $\mathcal{ASYNC}$ scheduler is more general and realistic among all the scheduler models.  

In any model, the robots can be considered as transparent or opaque. In the case of transparent robots, a robot can see another robot even if there are other robots between them. But in the case of opaque robots, a robot can not see another robot if there are other robots between them. There are many works where both these models have been considered (\cite{BoseAKS20,BoseKAS21,BoseAKS20grid,FlocchiniPSW08,absGridOpaque,BramasT16,0001FSY15,AdhikaryKS21,BhagatCM18} ). 
Opaque robots can be considered to be dimensionless (i.e point robots) (\cite{BoseKAS21,BoseAKS20grid,FlocchiniPSW08,FelettiMP18,AdhikaryKS21,BhagatCM18,abstriangulargidgather,CohenP08,LukovszkiH14}) or they can have some dimension (i.e fat robots) \cite{BoseAKS20}.
 In the literature on $\mathcal{APF}$, there are many works which have considered the robots to be dimensionless (i.e point robots) and opaque \cite{AdhikaryKS21,BoseKAS21}. But in practical application, it is impossible to design a point robot as any physical object must have some dimensions. So in our work, we have considered the robots to have some dimension. In fact, we have considered the robots to be a disc of radius `$rad$\rq, where $ rad \le \frac{1}{2}$.
 
In this paper, we are interested in the problem of $\mathcal{APF}$ on an infinite grid where the robots are considered to be fat and opaque and are placed on distinct vertices of the grid. The goal is to design an algorithm  $\mathcal{A}$ such that the robots after executing $\mathcal{A}$ form a pattern which is provided to each of the robots as input. In this paper, we have provided such an algorithm that solves the problem of $\mathcal{APF}$ under an $\mathcal{ASYNC}$ scheduler. 
\subsection{Earlier Works}
The arbitrary pattern formation problem was introduced first in \cite{Suzuki96distributedanonymous} and it has become a popular topic for research. It has been vastly studied under different types of environments and different types of settings (\cite{BoseAKS20,BoseKAS21,BoseAKS20grid,FlocchiniPSW08,absGridOpaque,BramasT16,0001FSY15,FelettiMP18,AdhikaryKS21,LukovszkiH14,BramasT18,CiceroneSN19,CiceroneSN19a,DieudonnePV10}). In most of these works, the basic assumption was that the robots are points and they do not have obstructed visibility. But in a practical application-based scenario designing a point robot is impossible because every physical object has a certain dimension. So in \cite{BoseAKS20}, authors have considered a swarm of fat and opaque robots and shown that this swarm can form any given pattern from any asymmetric configuration without collision under the $\mathcal{LUMI}$ model using 10 colours in a plane. This luminous model was first introduced in \cite{Peleg05} by Peleg. The visible lights can be used as a means of communication and persistent memory.
Designing a collision-free algorithm in plane is easier than handling collision in discrete domain. This is because, in plane the robots can move freely in any direction avoiding other robots but in discrete domain, there can be only one single path to reach from one point to another point. This is why many researchers became interested to study the problem of $\mathcal{APF}$ in discrete domain. In \cite{BoseAKS20grid}, an algorithm for $\mathcal{APF}$ has been provided for a swarm of point robots on an infinite grid but considering full and Unobstructed visibility. Now in \cite{AdhikaryKS21}, considering obstructed visibility model the authors have shown that a circle can be formed on an infinite grid from any initial configuration if the opaque point robots in the swarm have one-axis agreement and 7 colours. Then, in \cite{absGridOpaque} authors have presented another algorithm where a swarm of opaque point robots on an infinite grid can form the given pattern in finite time using one-axis agreement and 8 colours. But none of these works considered fat robots on infinite grid and solve the problem of $\mathcal{APF}$.

To the best of our knowledge, there is no work till now which has considered fat robots on infinite grid and provided any algorithm for arbitrary pattern formation on the grid. So, in this paper, we have considered a swarm of opaque fat robots on an infinite grid and provided an algorithm (\textsc{ApfFatGrid}) where the swarm can form a predefined given pattern on the grid using 9 colours.

\subsection{Problem description and our contribution}
This paper deals with the problem of arbitrary pattern formation on an infinite grid using luminous opaque fat robots with 9 colours. The robots are considered to be a disk having a fixed radius `$rad$\rq, which is less or equal to $\frac{1}{2}$. The robots manoeuvre in a \textsc{Look-Compute-Move} (LCM) cycle under an adversarial asynchronous scheduler. The robots are autonomous, anonymous, identical and homogeneous. The robots only move to one of its four adjacent grid points and their movement is considered to be instantaneous (i.e a robot can only be seen on a grid point). The robots have one-axis agreement. Here, it is assumed that the robots do not agree upon any global coordinate though all robots agree on the direction and orientation of the $x$-axis . Initially, the centre of each robot is on a grid point of the infinite grid and a target pattern is provided to each of them. The robots are needed to agree on a global coordinate system and embed the target pattern according to the global coordinate and then move to the target locations to form the target pattern.

The main difficulty of $\mathcal{APF}$ lies in the problem of Leader Election problem. For that, the initial configuration is assumed to be asymmetric or there is at least one robot on some line of symmetry. Even with this assumption, it is quite hard to elect a leader as the vision of the robots becomes obstructed since the robots are opaque and fat. So, the main challenge of this problem is to elect a leader depending on the local view of each robot. The algorithm described in this paper does so. Another massive challenge of this problem is to avoid collision during the movement of robots on the grid. Our algorithm handles this by providing sequential movement of the robots and for this purpose $\mathcal{LUMI}$ model has been used. 

The problem, we have considered in this paper, is very practical in nature. Restricted movement, robots with dimension and obstructed visibility all these assumptions are very much practical in terms of designing robots. The algorithm presented in this paper solves the $\mathcal{APF}$ on infinite grid with a swarm of luminous, opaque and fat robots with finite time. A comparison table is provided below which will help readers to compare our work to the previous such works.

\vspace{0.01\linewidth}



\begin {table}[ht]

\begin{center}
 \begin{tabular}{ | m{4em} | m{2.5cm}| m{2.5cm}| m{2.3cm}| m{1.8cm}| m{3cm}| } 
\hline
 \textbf{Paper} & \textbf{Environment} & \textbf{Visibilty} &\textbf{Robot Type} & \textbf{\#Colours}\\
 \hline
  \cite{BoseAKS20grid} & Grid & Unobstructed & Point & 0\\
  \hline
    \cite{BoseKAS21}& Plane & Opaque & Point & 6 \\
  \hline

    \cite{absGridOpaque} & Grid & Opaque & Point & 8 \\
  \hline
  \cite{BoseAKS20}& Plane & Opaque & Fat & 10 \\
  \hline
  \textbf{This paper} &  Grid & Opaque & Fat & 9\\
  \hline
\end{tabular}
\end{center}
\end{table}

\section{Model and Definitions}\label{model}

\subsection{Model}

~~~~~\textbf{Grid:} The infinite two-dimensional grid $\mathcal{G}$ is a weighted graph $\mathcal{G} = (V,E) $ such that each node $v \in V$ has four adjacent nodes $v_0, v_1, v_2$ and $v_3 \in V$ and the edges $vv_{i \pmod {4}} \in E $ is perpendicular to the edge $vv_{i+1 \pmod{4}} \in E$. Also, the weight of each edge $e \in E$ is  basically the length of the edge $e$ which is considered to be 1 unit in this work.

~~\textbf{Robots:} In this work, a set of $n$ robots $R = \{r_0, r_1, \dots, r_{n-1}\}$ are considered to be autonomous, anonymous, homogeneous and identical. This means that the robots do not have any central control, they do not have any unique identifiers such as IDs and they are indistinguishable by their physical appearance. The robots are also considered to have some dimension i.e. the robots are considered to be a disk of radius `$rad$\rq ($rad \le \frac{1}{2}$) rather than points. The robots are deployed on a two-dimensional infinite grid $\mathcal{G}$, where each of them is initially positioned in such a way that their centre is on distinct grid points of $\mathcal{G}$. The robots are considered to have an agreement over the direction and orientation of $x$-axis i.e, all the robots have an agreement over left and right but the robots do not have any agreement over the $y$-axis. Also, they do not have knowledge of any global coordinate system other than their agreement over the direction of $x$-axis. 
Here in this paper, we have considered the robots to have light. A light of any robot can have $\mathcal{O}(1)$ distinct colours. A robot $r \in R$ can see the colour of its own light and the colour of the lights of other robots that are visible to $r$. In this work, we have assumed that the light of each robot has nine distinct colours namely \texttt{off}, \texttt{terminal1}, \texttt{candidate}, \texttt{call}, \texttt{moving1}, \texttt{reached}, \texttt{leader1}, \texttt{leader} and \texttt{done}.




\textbf{Look-Compute-Move cycles:} A robot $r \in R$, when active, operates according to the \textsc{Look-Compute-Move} (LCM) cycle. In the \textsc{Look} phase, a robot takes the snapshot of the configuration to get the positions represented in its own local coordinate system and the colours of the light of all other robots visible to it. Then, $r$ performs the computation phase where it decides the position of the adjacent grid point where it will move next and changes the colour of its light if necessary depending on the input it got from the \textsc{Look} phase. In the \textsc{Move} phase, $r$ moves to the decided grid point or makes a null move. The movements of robots are restricted only along grid lines from one grid point to one of its four adjacent grid points. The movements of robots are assumed to be instantaneous in discrete domain. Here, we assume that the movements are instantaneous i.e., they are always seen on grid points, not on edges. 

\textbf{Scheduler:} We assume that the robots are controlled by an asynchronous adversarial scheduler. That implies the duration of the three phases \textsc{Look}, \textsc{Compute} and \textsc{Move} are finite but unbounded. So, there is no common notion of round for this asynchronous scheduler.

\textbf{Visibility:} The visibility of robots is unlimited but by the presence of other robots it can be obstructed. A robot $r_i$ can see another robot $r_j$ if and only there is a point $p_{r_j}$ on the boundary of $r_j$ and $p_{r_i}$ on the boundary of $r_i$ such that the line segment $\overline{p_{r_i}p_{r_j}}$ does not intersect with any point occupied by other robots in the configuration. Now, it follows from the definition that $r_i$ can see $r_j$ implies $r_j$ can see $r_i$.

\textit{\textbf{Configuration:}} We assume that the robots are placed on  the infinite two-dimensional grid $\mathcal{G}$. Next we define a function $f:V \rightarrow \{0,1\}$, where $f(v)$ is the number of robots placed on a grid point $v$. Then $\mathcal{G}$ together with the function $f$ is called a configuration which is denoted by $ \mathbb{C}=(\mathcal{G},f)$. For any time $T$, $\mathbb{C}(T)$ will denote the configuration of the robots at $T$.

\subsection{Notations and Definitions}
We have used some notations throughout the paper. A list of these notations is mentioned in the following table.

\vspace{0.01\linewidth}
\begin{center}
\begin{tabular}{ | m{4em} | m{10cm}| } 
\hline
$\mathcal{L}_1$& First vertical line on left that contains at least one robot.\\
  \hline
     $\mathcal{L}_V(r)$ & The vertical line on which the robot $r$ is located.\\
    \hline
    $\mathcal{L}_H(r)$ & The horizontal line on which the robot $r$ is located.\\
    \hline
     $\mathcal{L}_I(r)$ & The left immediate vertical line of robot $r$ which has at least one robot on it. \\
   \hline
     $\mathcal{R}_I(r)$&  The right immediate vertical line of robot $r$ which has at least one robot on it. \\
    \hline
    ${H}_L^O(r)$ &  Left open half for the robot $r$. \\
    \hline
     ${H}_L^C(r)$ & Left closed half for the robot $r$ (i.e  ${H}_L^O(r) \cup \mathcal{L}_V(r)$).\\
     \hline
     ${H}_B^O(r)$ &  Bottom open half for the robot $r$.\\
     \hline
     ${H}_B^C(r)$ &  Bottom closed half for the robot $r$ (i.e  ${H}_B^O(r) \cup \mathcal{L}_H(r)$).\\
     \hline
     ${H}_U^O(r)$ & Upper open half for the robot $r$.\\
    \hline
    ${H}_U^C(r)$ & Upper closed half for the robot $r$ (i.e  ${H}_U^O(r) \cup \mathcal{L}_H(r)$). \\
    \hline
    $K$ & The horizontal line passing through the middle point of the line segment between two robots with light \texttt{candidate} or \texttt{call} or \texttt{reached} on the same vertical line. \\
    \hline
    $l_{next}(r)$ & The next vertical line on the right of $\mathcal{L}_V(r)$.\\
    \hline
    $\mathcal{H}_{last}$ & The lowest horizontal line having a robot with colour \texttt{done}.\\
    \hline
    
\end{tabular}
\end{center}

\textit{\textbf{Terminal Robot:}} A robot $r$ is called a terminal robot if there is no robot below or above $r$ on $\mathcal{L}_V(r)$.

\textit{\textbf{Symmetry of a vertical line $L$ w.r.t $K$:}} Let $\lambda$ be a binary sequence defined on a vertical line $L$ such that $i$-th term of $\lambda$ is defined as follows:
 \begin{equation*}
 \lambda(i) = \begin{cases}
       1 & \text{if $\exists$ a robot on the $i$-th grid point from $K \cap L$ on the line $L$.} \\
       0 & \text{otherwise.}
     \end{cases}
\end{equation*}
Since there are two $i$-th grid points from $K \cap L$ on the line $L$ (above $K$ and below $K$), there are two such values of $\lambda$, say $\lambda_1$ and $\lambda_2$. If $\lambda_1 = \lambda_2$, then $L$ is said to be symmetric with respect to $K$. Otherwise, it is said to be asymmetric with respect to $K$. Henceforth, whenever the symmetry of a line is mentioned, it means the symmetry of the line with respect to $K$.

\textit{\textbf{Dominant half:}} A robot $r$ is said to be in the dominant half if for $\lambda_1 > \lambda_2$ (lexicographically) on $\mathcal{R}_I(r)$, $r$ and the portion of $\mathcal{R}_I(r)$ corresponding to $\lambda_1$ lie on same half-plane delimited by $K$.

\vspace{0.01\linewidth}

\section{The Algorithm}
The main result of the paper is Theorem \ref{thm1.1}. The proof of the ‘only if’ part is the same as in the case for point
robots, proved in \cite{BoseKAS21}. The ‘if’ part will follow from the algorithm presented in this section.

\begin{theorem}
\label{thm1.1}
For a set of opaque fat robots having one-axis agreement, $\mathcal{APF}$ is deterministically solvable if and only if the initial configuration is not symmetric with respect to a line $K$ such that 1) $K$ is parallel to the agreed axis and 2) $K$ is not passing through any robot.
\end{theorem}

For the rest of the paper, we shall assume that the initial configuration $\mathbb{C}(0)$ does not admit the
unsolvable symmetry stated in Theorem \ref{thm1.1}. Our Algorithm executes in two phases. In the first phase, a leader is elected and in the second phase, the robots form the target pattern embedded on the grid using the location of the leader as an agreement to the origin of a global coordinate system. The phases are described in detail in the following subsections.

\subsection{Phase 1}
Initially, at $\mathbb{C}(0)$ all the robots are on the grid $\mathcal{G}$ with colour \texttt{off}. Note that in $\mathbb{C}(0)$, there are at least one and at most two terminal robots on $\mathcal{L}_1$. These robots change their colours to \texttt{terminal1}. A robot with colour \texttt{terminal1} changes its colour to \texttt{candidate} and moves if it sees it has its left open half empty. Also, if a robot $r$ with colour \texttt{candidate} is a singleton in $H_{L}^{C}(r)$ and all robots in $\mathcal{R}_I(r)$ are \texttt{off}, it changes its colour to \texttt{leader1}. Note that due to the asynchronous scheduler, it might happen that $r$ is a singleton in $H_{L}^{C}(r)$ with colour \texttt{candidate} and there is another robot $r'$ on $\mathcal{R}_I(r)$ with colour \texttt{terminal1}. In this case, if $r$ awakes, it does not change its colour to \texttt{leader1} as it does not see all robots on $\mathcal{R}_I(r)$ have colour \texttt{off}. Also if $r'$ awakes, it sees $r$ with colour \texttt{candidate} in $\mathcal{L}_I(r')$ and turns its colour to \texttt{off}. In this scenario, $r$ becomes singleton in $H_{L}^{C}(r)$ and sees all robots on $\mathcal{R}_I(r)$ have colour \texttt{off}. So, $r$ changes its colour to \texttt{leader1}. Now consider that both $r$ and $r'$ are on the same vertical line $\mathcal{L}_V(r)$ with colour  \texttt{candidate} such that there is no robot between $\mathcal{L}_H(r)$ and $\mathcal{L}_H(r')$. 
Note that in this configuration, all robots on $\mathcal{R}_I(r)$ (i.e. $\mathcal{R}_I(r')$) have colour \texttt{off}. In this case, both $r$ and $r'$ check the symmetry of $\mathcal{R}_I(r)$ with respect to the line $K$ (i.e the horizontal line which is equidistant from both $r$ and $r'$). If $\mathcal{R}_I(r)$ is not symmetric with respect to $K$, then one of $r$ or $r'$ whichever is on the dominant half changes the colour to \texttt{leader1}. On the other hand, if $\mathcal{R}_I(r)$ is symmetric, both the robots $r$ and $r'$ change their colours to \texttt{call} from \texttt{candidate}. Now all the robots on $\mathcal{R}_I(r)$ have colour \texttt{off} and all of them can see exactly two robots with colour \texttt{call} on their left immediate line. Note that since all the robots on $\mathcal{R}_I(r)$ can see both $r$ and $r'$, all of them also know the line $K$. Now if there is any robot on $K$, it changes its colour to \texttt{leader1}. Otherwise, the robots on $\mathcal{R}_I(r)$ (at least one and at most two robots) which are closest to $K$ change the colours to \texttt{moving1}. A robot with colour \texttt{off} on $\mathcal{R}_I(r)$ which is not closest to $K$, changes its colour to \texttt{moving1} when it sees another robot with colour \texttt{moving1} on the same vertical line. Note that after a finite time, at least  all robots either above or below $K$ which are on $\mathcal{R} _I(r)$ change their colours to \texttt{moving1} if $\mathcal{R}_I(r)$ is symmetric with respect to $K$.
Now suppose a robot with colour \texttt{moving1} say $r_1$, is terminal on $\mathcal{R}_I(r)$. Also, note that $r_1$ can see  at least one of $r$ and $r'$ on $\mathcal{L}_I(r_1)$. Now, if $r_1$ sees another robot $r_2$ on $\mathcal{L}_V(r_1)$ and no robot with colour \texttt{reached} on $\mathcal{L}_I(r_1)$, it moves vertically  opposite to $r_2$. Otherwise, if it is singleton on $\mathcal{L}_V(r_1)$ and sees no robot with colour \texttt{reached} on $\mathcal{L}_I(r_1)$, it moves vertically according to its positive $y-$axis until there is no robot either in $H_{U}^{C}(r_1) \cap \mathcal{L}_I(r_1)$ or in $H_{B}^{C}(r_1) \cap \mathcal{L}_I(r_1)$ and then towards left until it reaches $\mathcal{L}_V(r)$. Now when $r$ or $r'$ with colour \texttt{call} sees all robots on $\mathcal{R}_I(r)$ have colour \texttt{off} and sees a robot with colour \texttt{moving1} or \texttt{reached} on the same vertical line, then it changes its colour to \texttt{reached}. Due to the asynchronous environment, it might happen that one of $r$ or $r'$ does not see a robot with colour \texttt{moving1} on the same vertical line, but it is guaranteed that it will see a robot with colour \texttt{reached} on the same vertical line after a finite time. So, $r$ or $r'$ can change their colours to \texttt{reached} if all robots on $\mathcal{R}_I(r)$ have colour \texttt{off} and there is a robot on $\mathcal{L}_V(r)$ with colour \texttt{reached}. So in finite time, both the robots with colour \texttt{call} change their colours to \texttt{reached} (when there was no robot on $K \cap \mathcal{R}_I(r)$). Now a robot say $r_3$ with colour \texttt{moving1} on $\mathcal{L}_V(r)$ moves to the left if all robots on $\mathcal{R}_I(r_3)$ are with colour \texttt{off} and it can see a robot with colour \texttt{reached} on $\mathcal{L}_V(r_3)$. Due to asynchrony, it may happen that $r$ and $r'$ changed their colours to \texttt{reached} and after that a robot say $r_4$, on $\mathcal{R}_I(r)$ changes its colour to \texttt{moving1}. Observe that in this case, the robot $r_4$ changes its colour to \texttt{off} whenever it sees at least one robot with colour \texttt{reached} on $\mathcal{L}_I(r_4)$, otherwise the robots with colour \texttt{moving1} on $\mathcal{L}_V(r)$ will not move left. So, after a finite time, all robots with colour \texttt{moving1} on $\mathcal{L}_V(r)$ move to $\mathcal{L}_1$ and at this moment $r$ and $r'$ will be the only two robots with colour \texttt{reached} on $\mathcal{L}_V(r)$ that are terminal also. In this situation, $r$ and $r'$ change their colours to \texttt{candidate}. Now for asynchrony, it may happen that $r$ and $r'$ changed their colours to \texttt{candidate} and after that a robot say $r_5$, on $\mathcal{R}_I(r)$ changes its colour to \texttt{moving1}. In this case, the robot $r_5$ changes its colour to \texttt{off} whenever it sees at least one robot with colour \texttt{candidate} on $\mathcal{L}_I(r_5)$. Therefore, then $ r$ and $r'$ are with colours \texttt{candidate} and all robots on $\mathcal{R}_I(r)$ have colour \texttt{off}. So, they again check the symmetry of the new $\mathcal{R}_I(r)$ repeating the whole process. Thus after a finite time, a robot with colour \texttt{off} or a robot $r$ or $r'$ with colour \texttt{candidate} whoever is on dominant half changes its colour to \texttt{leader1}. A robot say $r_l$ with colour \texttt{leader1} always moves to the left when it sees other robot in $H_{L}^{C}(r_l)$ or $l_{next}(r_l)$, no robot with colour \texttt{call} on $\mathcal{L}_I(r_l)$ and no robot with colour \texttt{candidate} on $\mathcal{L}_V(r_l)$. 
\begin{algorithm}
\caption{\textsc{ApfFatGrid}: Phase 1}
  \label{leader_selection}
\LinesNumbered
    \SetKwProg{Pr}{Procedure}{}{}
\SetAlgoLined
       \footnotesize                        
  \Pr{\textsc{Phase1()}}{
  
  $r \leftarrow$ myself
  
   \SetAlgoVlined \uIf{$r.light =$ \texttt{off}}
         {
         
         \uIf{there is no robot in $H_{L}^{O}(r)$, no robot with light \texttt{leader1} in $\mathcal{R}_I(r) \cup \mathcal{L}_V(r)$ and $r$ is terminal on $\mathcal{L}_V(r)$}{$r.light \leftarrow$ \texttt{terminal1}}
         \uElseIf{there are exactly two robots in $\mathcal{L}_I(r)$ and their lights are \texttt{call} and $r$ is closest to $K$}{

                        \uIf{$r$ is on $K$}{$r.light =$ \texttt{leader1} }
                        
                        \Else{$r.light =$ \texttt{moving1}}

                                                 }
        \ElseIf{there is a robot with light \texttt{moving1} in $\mathcal{L}_V(r)$ }{
                            $r.light =$ \texttt{moving1}
                            }

         }
         
         \SetAlgoVlined \uElseIf{$r.light =$ \texttt{terminal1}} 
             {
             
             \uIf{there is no robot in $H_{L}^{O}(r)$}{$r.light \leftarrow$ \texttt{candidate}\\ move left}
             
             \ElseIf{there is a robot with light \texttt{candidate} in $\mathcal{L}_I(r)$}{$r.light =$ \texttt{off}}
             
             }
    
    \SetAlgoVlined \uElseIf{$r.light =$ \texttt{candidate}}
         {
         
         \uIf{$r$ is singleton in $H_{L}^{C}(r)$ and all robots in $\mathcal{R}_I(r)$ are \texttt{off} }{$r.light \leftarrow$ \texttt{leader1}}
         
         \ElseIf{there is a robot with light \texttt{candidate} or \texttt{call} on $\mathcal{L}_V(r)$, $r$ is terminal on $\mathcal{L}_V(r)$ and all robots in $\mathcal{R}_I(r)$ are \texttt{off}}{

                        \uIf{$\mathcal{R}_I(r)$ is symmetric with respect to $K$}{$r.light =$ \texttt{call} }
                        
                        \Else{\uIf{$r$ is in the dominant half}{}$r.light =$ \texttt{leader1}}
             
                                                                          }

         \ElseIf{there is a robot with light \texttt{leader1} on $\mathcal{L}_V(r)$}{$r.light =$ \texttt{off}}

         } 
        \SetAlgoVlined \uElseIf{$r.light =$ \texttt{moving1}}
         {
         
         \uIf{there is at least one robot with light \texttt{call} and no robot with light \texttt{reached} in $\mathcal{L}_I(r)$ and $r$ is terminal on $\mathcal{L}_V(r)$}{

                        \uIf{there is other robot both in $H_{U}^{C}(r) \cap \mathcal{L}_I(r)$ and $H_{B}^{C}(r) \cap \mathcal{L}_I(r)$ }{
                            \uIf{there is a robot $r'$ on $\mathcal{L}_V(r)$}{move opposite to $r'$}
                        \Else{move according to its positive $y-$axis}
                        }
                        
                        \Else{move left}

                                                                      }
         \ElseIf{there is a robot with light \texttt{reached} on $\mathcal{L}_V(r)$ and all robots in $\mathcal{R}_I(r)$ are \texttt{off}}{move left}                                             
          \ElseIf{there is at least one robot with light \texttt{reached} or \texttt{candidate} in $\mathcal{L}_I(r)$}{$r.light =$ \texttt{off}}    
    
         }  
  }
\end{algorithm}
\newpage

\begin{algorithm}
\setcounter{AlgoLine}{41}
\SetKwBlock{Begin}{}{end}
\footnotesize
\Begin{
    \uElseIf{$r.light =$ \texttt{call}}{
                                        \uIf{there is a robot with light \texttt{moving1} or, \texttt{reached} on $\mathcal{L}_V(r)$ and all robots in $\mathcal{R}_I(r)$ are \texttt{off}}{$r.light =$ \texttt{reached}}
                                        
                                        \ElseIf{there is a robot with light \texttt{leader1} in $\mathcal{R}_I(r)$}{$r.light =$ \texttt{off}}

                                        }
    
    \uElseIf{$r.light =$ \texttt{reached}}{
                                        \If{there is a robot with light \texttt{reached} or \texttt{candidate} on $\mathcal{L}_V(r)$, $r$ is terminal on $\mathcal{L}_V(r)$ and all robots in $\mathcal{R}_I(r)$ are \texttt{off}}{$r.light =$ \texttt{candidate}}
    
                                          }

    \ElseIf{$r.light =$ \texttt{leader1}}
        {

        \uIf{there is other robot in $H_{L}^{C}(r)$ or $l_{next}(r)$, no robot with light \texttt{call} in $\mathcal{L}_I(r)$ and no robot with light \texttt{candidate} on $\mathcal{L}_V(r)$}
                   {move left
                   }
        \Else{\uIf{there is other robot both in $H_{U}^{C}(r)$ and $H_{B}^{C}(r)$}{move vertically according to its positive $y-$axis}
              \Else{$r.light =$ \texttt{leader}}

             }

        }
  }
\end{algorithm}
Note that, a robot with colour \texttt{call} changes its colour to \texttt{off} if it sees a robot with colour \texttt{leader1} on its right immediate line. Also, a robot with colour \texttt{candidate} changes the colour to \texttt{off} when it sees a robot with colour \texttt{leader1} on the same vertical line. $r_l$ moves to the left until it becomes the singleton robot on the leftmost line of the configuration and there is no robot on $l_{next}(r_l)$ and then moves according to its positive $y-$axis until either one of $H_{U}^{C}(r_l)$ and $H_{B}^{C}(r_l)$ has no other robot. Note that it may happen due to the asynchronous environment that another robot with colour \texttt{candidate} moves to $l_{next}(r_l)$ while $r_l$ is on $\mathcal{L}_1$. In this case, when $r_l$ activates again it finds out it has non-empty $l_{next}(r_l)$ and moves left again even it was moving vertically in the previous activation. In this situation, when $r_l$ reaches a point where either one of $H_{U}^{C}(r_l)$ and $H_{B}^{C}(r_l)$ has no other robot, it changes its colour to \texttt{leader} and \textit{Phase 1} ends.

The following Theorem \ref{fth1} and Lemmas \ref{flemma1}$-$\ref{last_ph1} justify the correctness of the Algorithm 1.

\begin{theorem}
\label{fth1}
For any initial configuration $\mathbb{C}(0)$, $\exists$ $T > 0$ such that $\mathbb{C}(T)$ have exactly two robots with light \texttt{candidate} or exactly one robot with light \texttt{leader1} in $L_1$.
\end{theorem}
\begin{proof}
Observe that there can be at least one and at most two robots in $\mathbb{C}(0)$ such that they have their left open half empty and are terminal on $\mathcal{L}_1$.
Let there is only one robot $r_1$, who has $H_L^O(r_1)$ empty and is terminal on $\mathcal{L}_1$ (Figure \ref{fig:sigletonTerminal1}). This implies $r_1$ is singleton on $\mathcal{L}_1$. In this case, $r_1$ changes its colour to \texttt{terminal1} at some time $T' > 0$ and eventually changes to \texttt{leader1} at a time $T > T'$.  
\begin{figure}[ht]
    \centering
    \includegraphics[height=6cm,width=6cm]{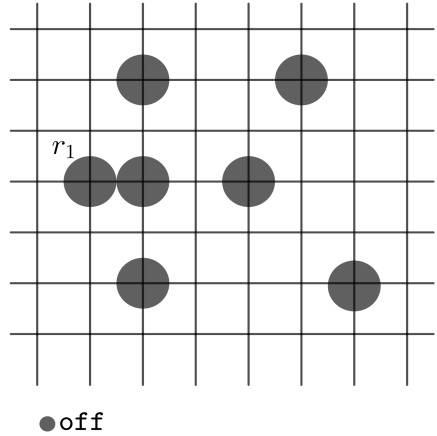}
    \caption{$r_1$ is singleton robot on $\mathcal{L}_1$.}
    \label{fig:sigletonTerminal1}
\end{figure}

Now let us consider the case where there are two robots $r_1$ and $r_2$ such that both $r_1$ and $r_2$ are terminal on $\mathcal{L}_1$ in $\mathbb{C}(0)$. Now if any one of $r_1$ or $r_2$ awakes, it changes its colour to \texttt{terminal1}. A robot with colour \texttt{terminal1} moves left after changing its colour to \texttt{candidate} if it has its left open half empty. Due to asynchronous environment, the following cases may occur.

\textbf{Case-I:} Let us consider the case where $r_1$ already changed its colour to \texttt{candidate} from \texttt{terminal1} and moved to $\mathcal{L}_1$ at a time $T_1 > 0$ and $r_2$ wakes after $T_1$ (Figure \ref{fig:Lemma1case1}). Then $r_2$ remains with colour \texttt{off} as it sees it is not on $\mathcal{L}_1$ anymore. Then $r_1$ during the next activation sees it is singleton on $H_L^C(r_1)$ and all robots on $\mathcal{R}_I(r_1)$ have colour \texttt{off}. So, it changes its colour to \texttt{leader1}. 
\begin{figure}[!htb]\centering
   \begin{minipage}{0.45\textwidth}
    \includegraphics[height=6cm,width=6cm]{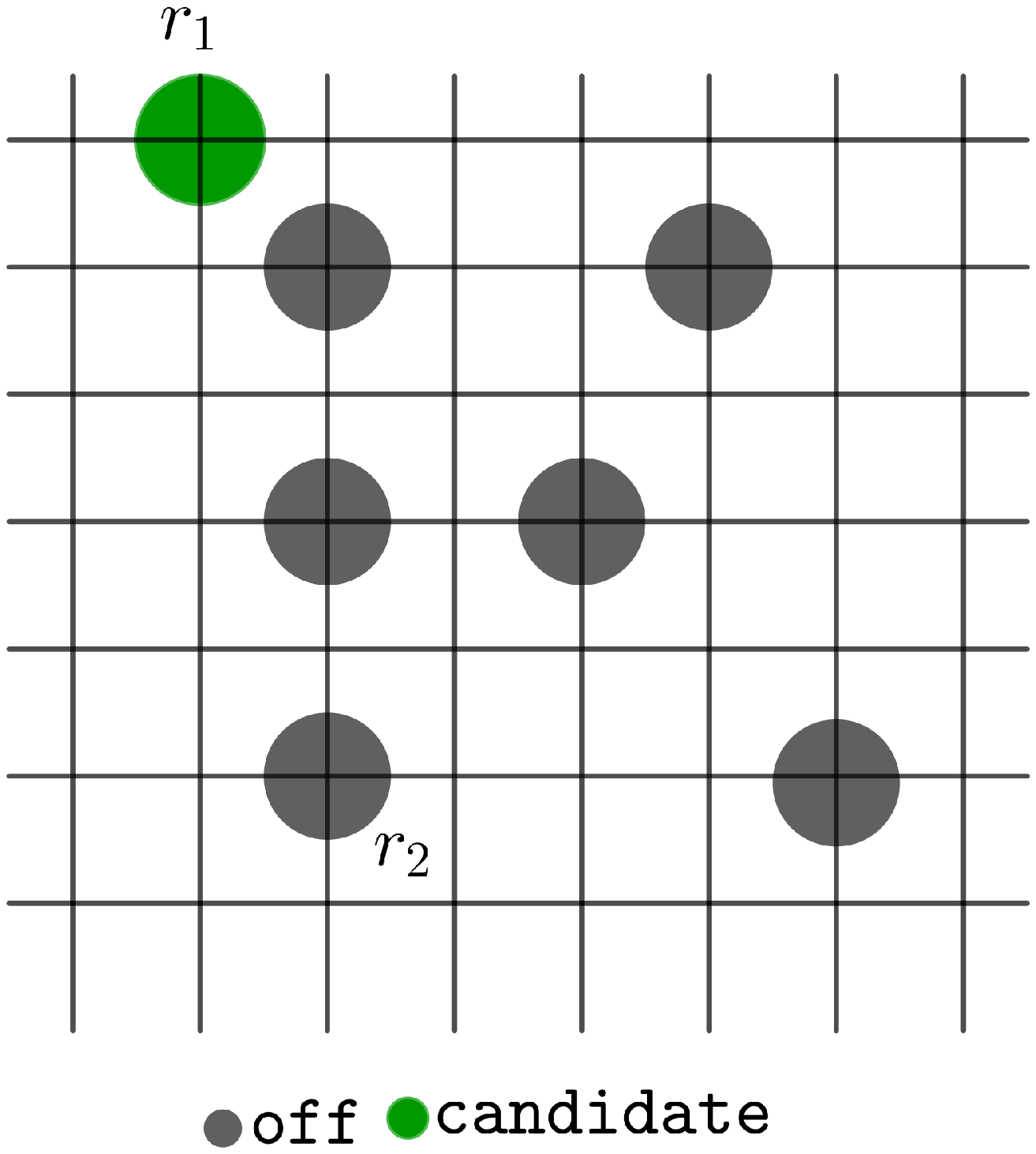}
     \caption{$r_1$ changes its colour to \texttt{candidate} and moves to $\mathcal{L}_1$  at time $T_1$ and $r_2$ wakes after time $T_1$.}\label{fig:Lemma1case1}
   \end{minipage}
   \hfill
   \begin {minipage}{0.45\textwidth}
    \includegraphics[height=6cm,width=6cm]{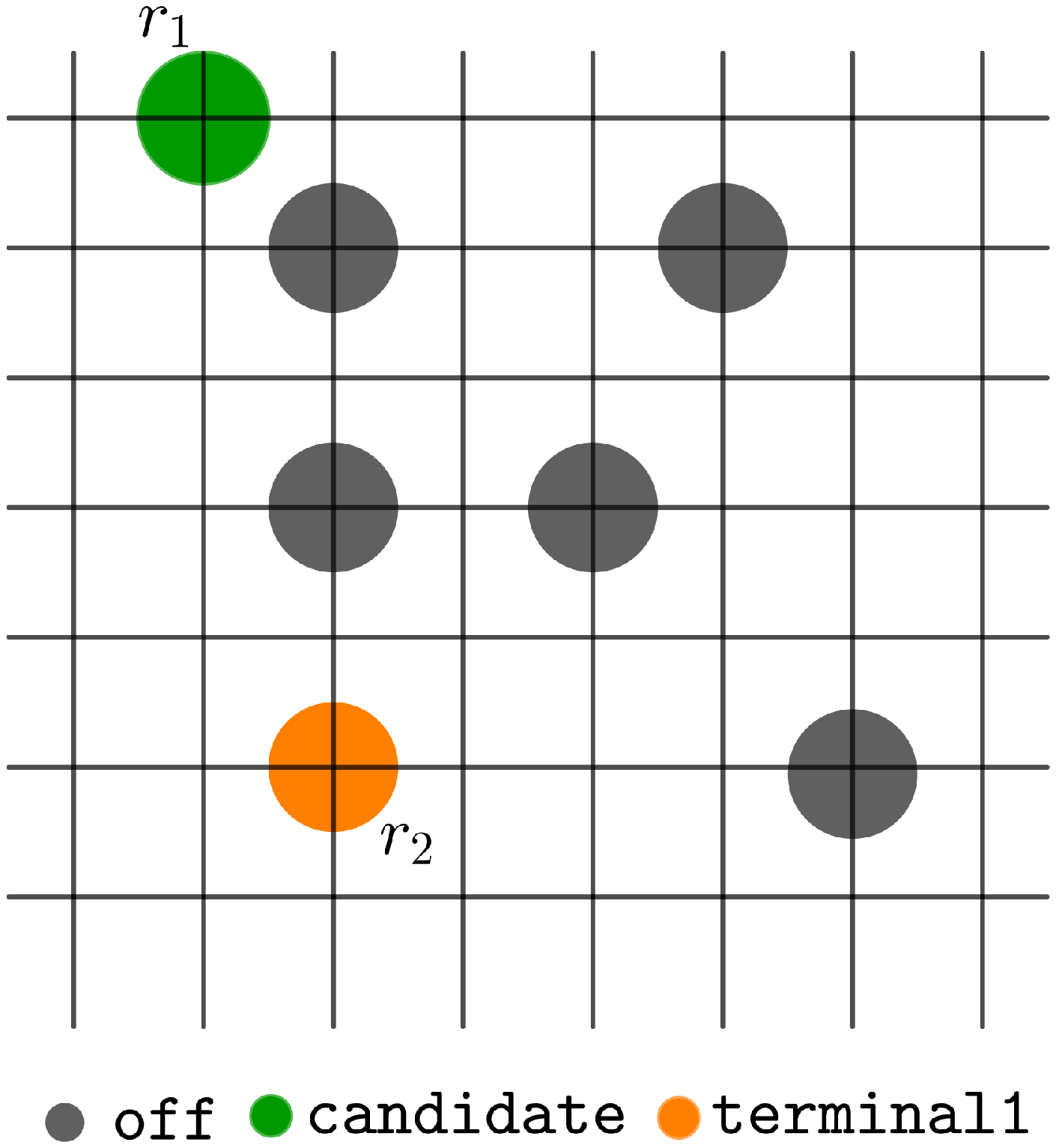}
     \caption{$r_1$ changed its colour to \texttt{candidate} and moves to $\mathcal{L}_1$ at time $T_1$ and $r_2$ changes its colour to \texttt{terminal1} at time $T_2 \ge T_1$. }\label{Fig:Lemma1case2}
   \end{minipage}
\end{figure}

\textbf{Case-II:} Let us consider the case where $r_1$ already changed its colour to \texttt{candidate} from \texttt{terminal1} and moved to $\mathcal{L}_1$ at a time $T_1 > 0$ and $r_2$ wakes  before $T_1$ and changes its colour to \texttt{terminal1} at a time $T_2 \ge T_1$ (Figure \ref{Fig:Lemma1case2}).  Now if again $r_1$ wakes between the times $T_1$ and $T_2$ (in this scenario $T_1 > T_2$), then it sees all robots on $\mathcal{R}_I(r_1)$ have colour \texttt{off} and $r_1$ is singleton on $H_L^C(r_1)$. So, $r_1$ changes its colour to \texttt{leader1} and $r_2$ does not change its colour as $H_L^C(r_2)$ has other robots. Now if $r_1$ wakes at a time $T_3$ where $T_3 > T_2$ and $r_2$ has not woke again, then it does not change its colour to \texttt{leader1} as it sees $r_2$ with colour \texttt{terminal1} on $\mathcal{R}_I(r_1)$. Now when $r_2$ wakes again at a time $T_4 > T_2$, it changes its colour to \texttt{off} as it sees $r_1$ with colour \texttt{candidate} on $\mathcal{L}_I(r_2)$. Now when $r_1$ wakes after $T_4$ again, it sees it is singleton on $H_L^C(r_1)$ and have all robots with colour \texttt{off} on $\mathcal{R}_I(r_1)$ and so changes its colour to \texttt{leader1}.

\textbf{Case-III:} Let us consider the case where $r_1$ already changed its colour to \texttt{candidate} from \texttt{terminal1} and moved to $\mathcal{L}_1$ at a time $T_1 > 0$ and $r_2$ wakes  before $T_1$ and changes its colour to \texttt{terminal1} at a time $T_2 < T_1$. Now, let $r_2$ wakes again at a time $T_3$.

\textbf{Case-III(a):} Now if $T_3 > T_1$, then even if $r_1$ wakes again between $T_3$ and $T_1$, it sees $r_2$  with colour \texttt{terminal1} on $\mathcal{R}_I(r_1)$. So, it does not change its colour to \texttt{leader1}. Now $r_2$ at time $T_3$ wakes and sees $r_1$ with colour \texttt{candidate} on $\mathcal{L}_I(r_2)$ and so changes its colour to \texttt{off}. Now $\exists $ $T_4$ such that $r_1$ wakes at $T_4 > T_3$ and sees it is singleton on $H_L^C(r_1)$ and have all robots on $\mathcal{R}_I(r_1)$ with colour \texttt{off}. So, $r_1$ changes its colour to \texttt{leader1}.

\textbf{Case-III(b):} 

\begin{figure}[ht]
    \centering
    \includegraphics[height=5.3cm,width=6cm]{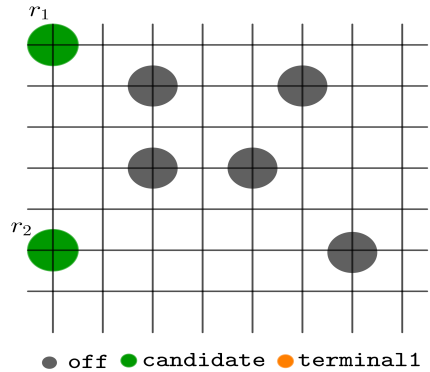}
    \caption{$r_2$ is with colour \texttt{terminal1} and $r_1$ changes its colour to \texttt{terminal1} and both $r_1$ and $r_2$ move to $\mathcal{L}_1$ changing their colour to \texttt{candidate at the same time.}}
    \label{fig:2candidate}
\end{figure}
Now if $T_3 = T_1$, then both $r_1$ and $r_2$ changes its colour to \texttt{candidate} and moves to $\mathcal{L}_1$. In this case, there will be two robots with colour \texttt{candidate} on $\mathcal{L}_1$ (Figure \ref{fig:2candidate}).
\textbf{Case-III(c):} Now when $T_3 < T_1$, it is similar as case-III(a).

Note that above all cases are exhaustive and in each case there is a time $T > 0$ such that in $\mathbb{C}(T)$, there is either one robot with colour \texttt{leader1} or two robots with colour \texttt{candidate} on $\mathcal{L}_1$.
\end{proof}

\begin{lemma}
\label{flemma1}
Any robot $r$ with colour \texttt{candidate} or, \texttt{call} or, \texttt{reached} always can see all robots with colour \texttt{off} in $\mathcal{R}_I(r)$ (if exist) and vice versa. 
\end{lemma}
\begin{proof}
Let us consider that there is exactly one robot $r$ and no other robot is on $\mathcal{L}_V(r)$. In this case, it is obvious that $r$ can see all robots including the robots with colour \texttt{off} on $\mathcal{R}_I(r)$ and vice versa.

Now let us consider that there are at least two robots $r$ and $r'$ on $\mathcal{L}_V(r)$ where colour of $r$ and $r'$ can be any one of \texttt{candidate}, \texttt{call} or, \texttt{reached} and $r$ is above $r'$. Now there are two cases.

\textbf{Case-I:} There are other empty vertical lines between $\mathcal{L}_V(r)$ and  $\mathcal{R}_I(r)$. In this case, let us take the common tangent $line_1$ of all robots on $\mathcal{R}_I(r)$ which is parallel to the line $\mathcal{R}_I(r)$ and nearest to $\mathcal{L}_V(r)$ and similarly take the common tangent $line_2$ of the robots on line $\mathcal{L}_V(r)$ parallel to $\mathcal{L}_V(r)$ and nearest to $\mathcal{R}_I(r)$. Let us denote the points where $line_1$ touches the terminal robots on $\mathcal{R}_I(r)$ as $p_1$ and $p_2$ respectively ($p_1$ is above $p_2$) and the points where $line_2$ touches $r'$ and $r$ as $p_3$ and $p_4$ respectively. Now let us draw a line segment say $line_3 = \overline{p_4p_1}$ and $line_4 = \overline{p_3p_2}$ . Observe that the area bounded by the lines $line_1, line_2, line_3$ and $line_4$ is a trapezoid which is a convex set containing no other robot (Figure \ref{fig:Lemma2case1}). Let $r_1$ be any robot with colour \texttt{off} on $\mathcal{R}_I(r)$. Let $line_1$ touches the robot $r_1$ at a point say $P$. Then the line segments $\overline{Pp_3}$ and $\overline{Pp_4}$ contains no robot on them. So, each of $r$ and $r'$ can see $r_1$. Thus $r$ and $r'$ can see all robots with colour \texttt{off} on $\mathcal{R}_I(r)$ and all robots with colour \texttt{off} can see both of $r$ and $r'$.

\begin{figure}[!htb]\centering
   \begin{minipage}{0.45\textwidth}
    \includegraphics[height=6cm,width=6cm]{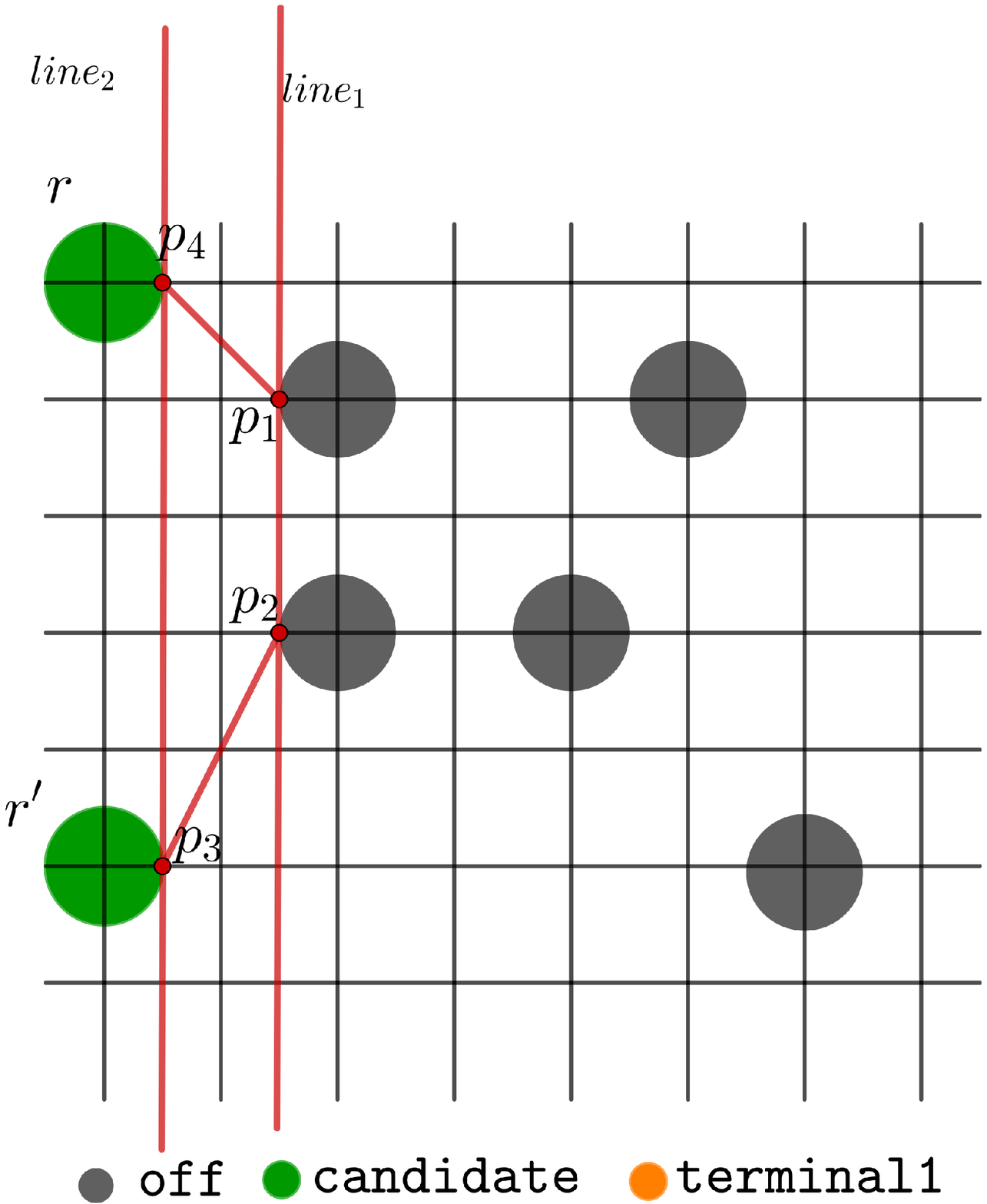}
     \caption{Area bounded by the quadrilateral  $\overline{p_1p_2p_3p_4}$ is convex. }\label{fig:Lemma2case1}
   \end{minipage}
   \hfill
   \begin {minipage}{0.45\textwidth}
    \includegraphics[height=6cm,width=6cm]{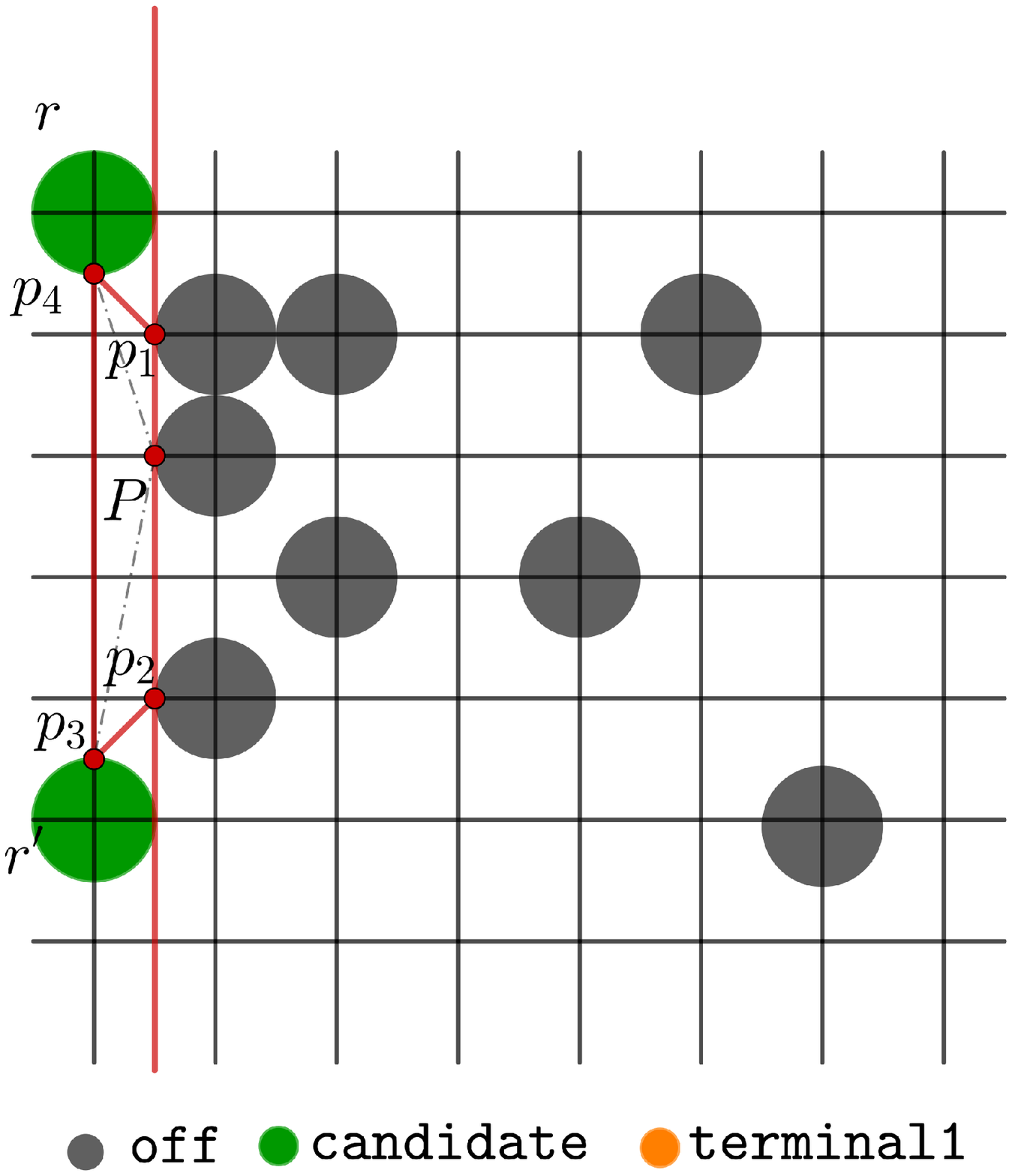}
     \caption{Area bounded by the quadrilateral  $\overline{p_1p_2p_3p_4}$ is convex. }\label{Fig:Lemma2case2}
   \end{minipage}
\end{figure}

\textbf{Case-II:} Next let there is no other vertical line between $\mathcal{L}_V(r)$ and $\mathcal{R}_I(r)$. Note that all robots with colour \texttt{off} must be between the lines $\mathcal{L}_H(r)$ and $\mathcal{L}_H(r')$. In this scenario, let us draw the common tangent $line_1$ of the robots on $\mathcal{R}_I(r)$  which is parallel to $\mathcal{R}_I(r)$ and nearest to $\mathcal{L}_V(r)$. Let us denote the points where $line_1$ touches the terminal robots on $\mathcal{R}_I(r)$ as $p_1$ and $p_2$ ($p_1$ is above $p_2$). Let us now denote the points where boundary of $r$ and $r'$ intersect the line $\mathcal{L}_V(r)$ which is nearest to $r'$ and $r$ respectively as $p_4$ and $p_3$. Let us call the line segment  $\overline{p_3p_4}$ as $line_2$, $\overline{p_4p_1}$ as $line_3$ and $\overline{p_3p_2}$ as $line_4$. Then the area bounded by these four lines is convex and there is no robot inside this area. 0For any robot $r_1$ on $\mathcal{R}_I(r)$ with colour \texttt{off}, let us denote the point where $line_1$ touches $r_1$ as $P$. Then both the line segment $\overline{Pp_3}$ and $\overline{Pp_4}$ do not contain any other robot (Figure \ref{Fig:Lemma2case2}). So, both $r$ and $r'$ can see $r_1$ and similarly $r_1$ sees both $r$ and $r'$. Thus $r$ and $r'$ can see all robots with colour \texttt{off} on $\mathcal{R}_I(r)$ and all robots with colour \texttt{off} can see both of $r$ and $r'$.

 So, we can conclude the lemma. 
\end{proof}
\begin{lemma}
If $r_1$ and $r_2$ be two robots with colour \texttt{call} or \texttt{reached} or \texttt{candidate} on the same vertical line, then any terminal robot $r$ with colour \texttt{moving1} on $\mathcal{R}_I(r_1)$ (= $\mathcal{R}_I(r_2)$) always can see at least one of $r_1$ and $r_2$.
\end{lemma}
\begin{proof}
Without loss of generality, let us assume that $r_1$ is above $r_2$ and $r$ is above $K$ (i.e the horizontal line which is equidistant from both $\mathcal{L}_H(r_1)$ and $\mathcal{L}_H(r_2)$). Also, let there is no other vertical line between $\mathcal{L}_I(r)$ and $\mathcal{L}_V(r)$, otherwise with the same argument as Lemma \ref{flemma1} we can say that $r$ can see both $r_1$ and $r_2$. Now there are three cases.

\textbf{Case-I:} $r$ is below $\mathcal{L}_H(r_1)$. In this case, by similar argument in Case-II of Lemma \ref{flemma1}, we can conclude that $r$ can see both $r_1$ and $r_2$. 

\textbf{Case-II:} $r$ is on $\mathcal{L}_H(r_1)$. Let us draw the tangents of $r$, $line_1$ parallel to $\mathcal{L}_V(r)$ and nearest to $\mathcal{L}_I(r)$ and tangent of $r_1$, $line_2$ parallel to $\mathcal{L}_V(r_1)$ and nearest to $\mathcal{R}_I(r_1)$. Now let $line_1$ touches $r$ at point $p_1$ and $line_2$ touches $r_1$ at point $p_2$ (Figure \ref{fig:Lemma3case2}). Since $\overline{p_1p_2}$  does not contain any other robot, $r$ can see $r_1$. Note that $line_1$ and $line_2$ can be same if the robots are of radius $\frac{1}{2}$. Now let $line_1$ touches both $r$ and $r_1$ at a point $p$. Hence $r$ can see $r_1$.

\begin{figure}[!htb]\centering
   \begin{minipage}{0.45\textwidth}
    \includegraphics[height=6.5cm,width=7cm]{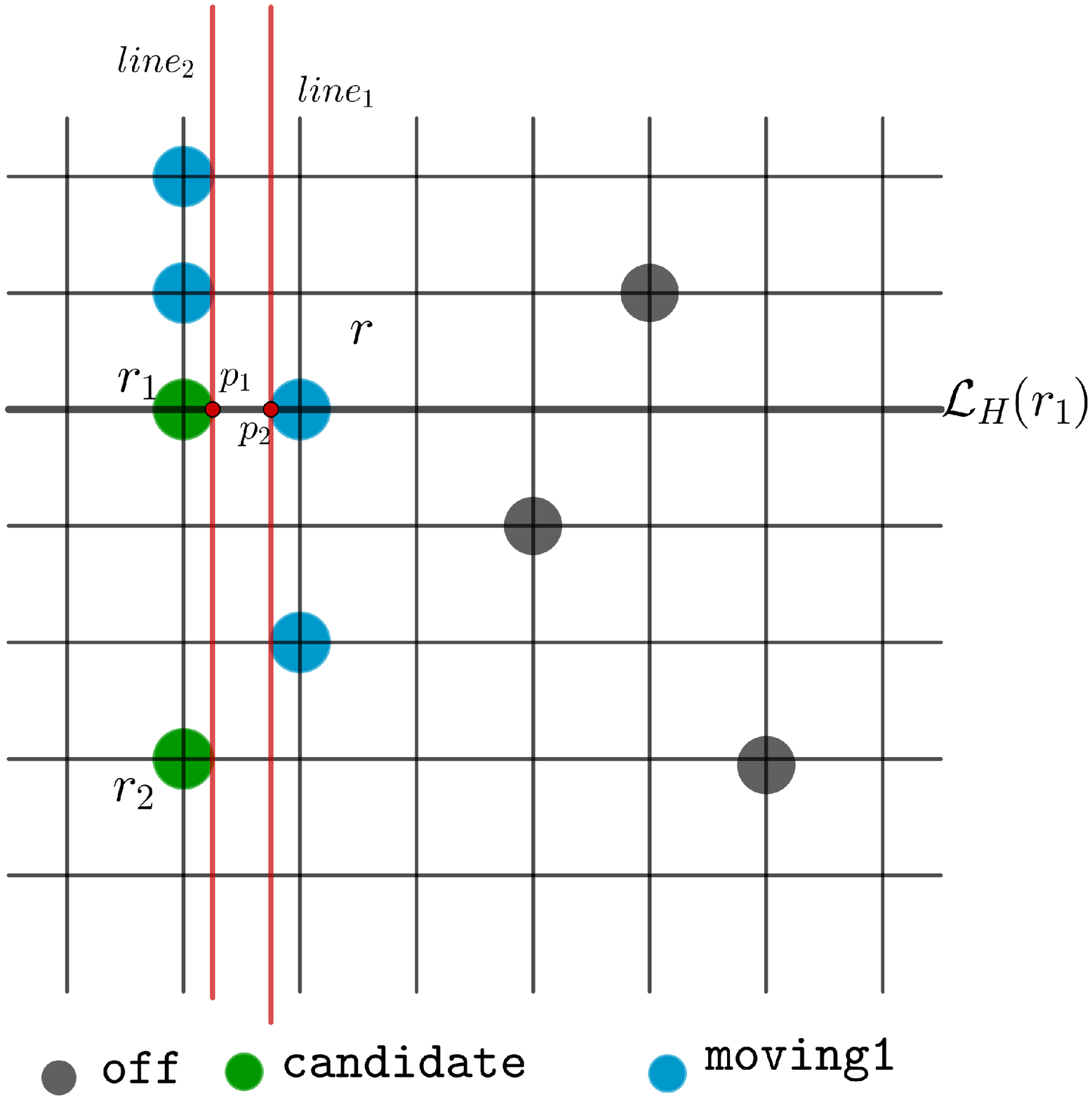}
     \caption{$r$ is on $\mathcal{L}_H(r_1)$.}\label{fig:Lemma3case2}
   \end{minipage}
   \hfill
   \begin {minipage}{0.45\textwidth}
    \includegraphics[height=6.5cm,width=6cm]{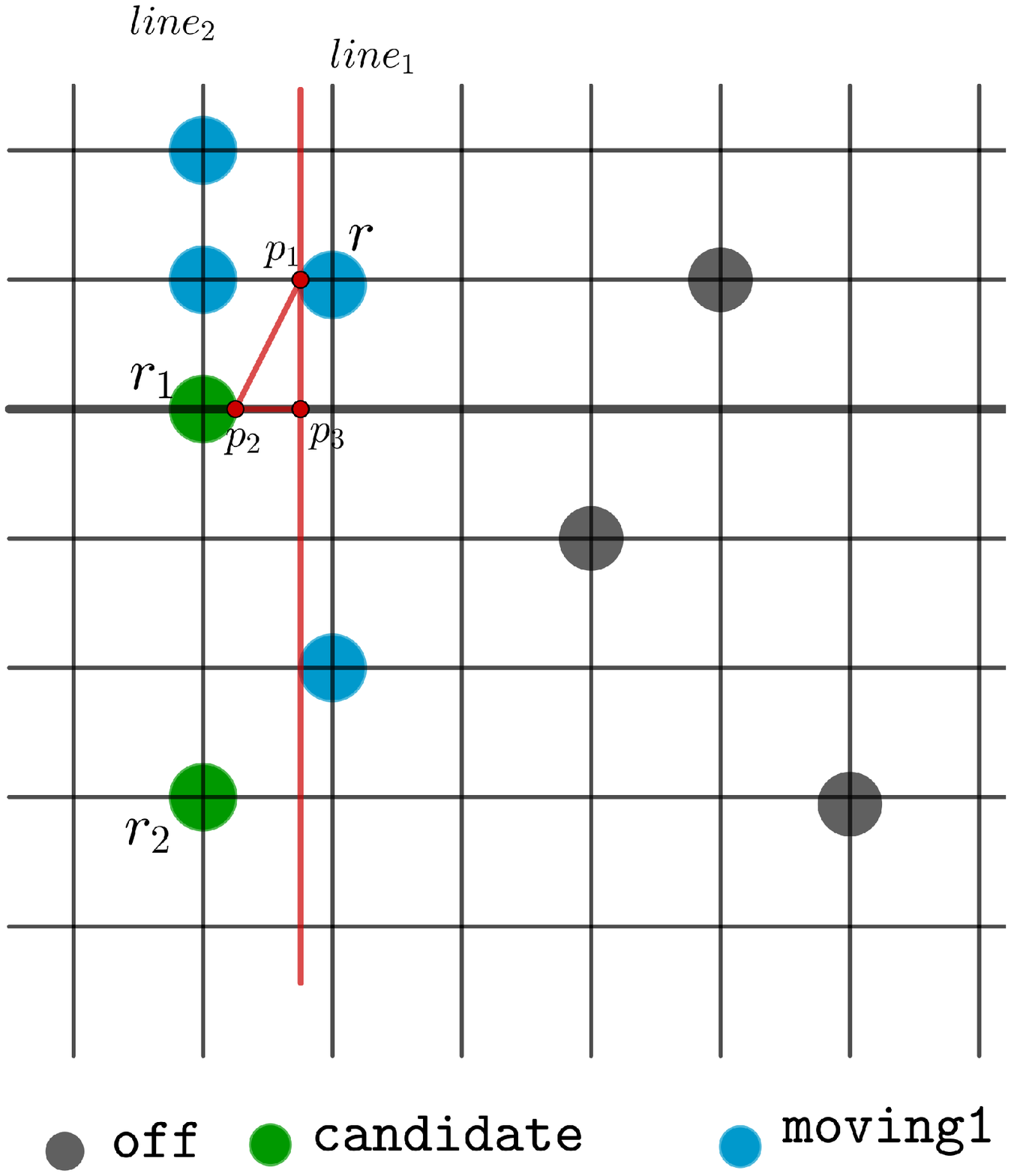}
     \caption{$r$ is above $\mathcal{L}_H(r_1)$.}\label{Fig:Lemma3case3}
   \end{minipage}
\end{figure}
\textbf{Case-III:} $r$ is above $\mathcal{L}_H(r_1)$. In this case, we claim that if there is any other robot with colour \texttt{moving1} on $\mathcal{L}_V(r)$, it must be below $\mathcal{L}_H(r_1)$. 
If possible let, there is another robot $r'$ on $\mathcal{L}_V(r)$ which is above $\mathcal{L}_H(r_1)$ but below $\mathcal{L}_H(r)$ with colour \texttt{moving1}. Since a robot turns its colour to \texttt{moving1} from \texttt{off}, there exists a time $T$ when $r'$ had colour \texttt{off}. So, in $\mathbb{C}(T)$, $r'$ must be located below $\mathcal{L}_H(r_1)$. Now, $r$ is already located on $\mathcal{L}_V(r')$ and above $r'$. So, if $r'$  is terminal, it moves opposite to $r$ and never reaches above $\mathcal{L}_H(r_1)$. And if $r'$ is not terminal, then it never moves until $r$ moves left. So, if $r$ is above $\mathcal{L}_H(r_1)$ with colour \texttt{moving1} and is terminal, then there is no other robot on the grid points on $\mathcal{L}_V(r)$ between $\mathcal{L}_H(r)$ and $\mathcal{L}_H(r_1)$. 
Let the tangent of $r$ which is parallel to $\mathcal{L}_V(r)$ and nearest to $\mathcal{L}_I(r)$ touches $r$ at point $p_1$ and intersects $\mathcal{L}_H(r_1)$ at $p_3$.
Also, boundary of $r_1$ touches the line $\mathcal{L}_H(r_1)$ at a point nearest to $\mathcal{L}_V(r)$ (say $p_2$) (Figure \ref{Fig:Lemma3case3}). Since $p_1$, $p_2$ and $p_3$ form a triangle and the area bounded by the triangle is a convex set containing no other robot, $r$ can see $r_1$.

\end{proof}

\begin{lemma}
\label{flemma2}
A robot changes its colour to \texttt{leader1} only from light \texttt{candidate} or \texttt{off}.
\end{lemma}
\begin{proof}
From Algorithm \ref{leader_selection}, it follows directly that a robot can change its colour to \texttt{leader1} only if it was either with colour \texttt{candidate} or with colour \texttt{off}.
\end{proof}
\begin{lemma}
A robot with light \texttt{leader1} always has  empty grid point in its left.
\end{lemma}
\begin{proof}
If at some tome $T>0$, there is only one robot say $r$, on $\mathcal{L}_1$ with colour \texttt{candidate} and no robot with colour other than \texttt{off} on $\mathcal{R}_I(r)$, then $r$  changes its colour to \texttt{leader1}. Observe that since $r$ is on $\mathcal{L}_1$, it will have its left grid point empty.

Now, consider there are two robots $r$ and $r'$ with colour \texttt{candidate} on $\mathcal{L}_V(r)$ (i.e. $\mathcal{L}_V(r')$). 
Now by Lemma \ref{flemma2}, it is evident that a robot with colour \texttt{candidate} or \texttt{off} can only change its colour to \texttt{leader1}. So, let us consider these cases.

\textbf{Case-I:} Let a robot $r_1$ with colour \texttt{off} changes its colour to \texttt{leader1}. That implies  only $r$ and $r'$ is on $\mathcal{L}_I(r_1)$ having colour \texttt{call} and $r_1$ is on $K \cap \mathcal{R}_I(r)$. Note that if $r$ and $r'$ are adjacent on $\mathcal{L}_V(r)$, then $K$ can  not be a horizontal line of the grid $\mathcal{G}$. So, $r$ and $r'$ are not adjacent on $\mathcal{L}_V(r)$ (Figure \ref{fig:Lemma5case1}). Now note that even if there are robots other than $r$ and $r'$ on $\mathcal{L}_V(r)$ or, $\mathcal{L}_I(r)$, they are not on or between the line $\mathcal{L}_H(r)$ and $\mathcal{L}_H(r')$. So, $\mathcal{L}_H(r_1)$ lies between $\mathcal{L}_H(r)$ and $\mathcal{L}_H(r')$ and $r_1$ is on $\mathcal{R}_I(r)$. So, we can say that $\mathcal{L}_H(r_1) \cap H_L^O(r_1)$ is always empty. As $r_1$ moves only on left until it becomes singleton on $\mathcal{L}_1$ and has $l_{next}(r_1)$ empty, $r_1$, the robot with light \texttt{leader1} always has its left grid point empty.

\begin{figure}[!htb]\centering
   \begin{minipage}{0.45\textwidth}
    \includegraphics[height=6cm,width=7cm]{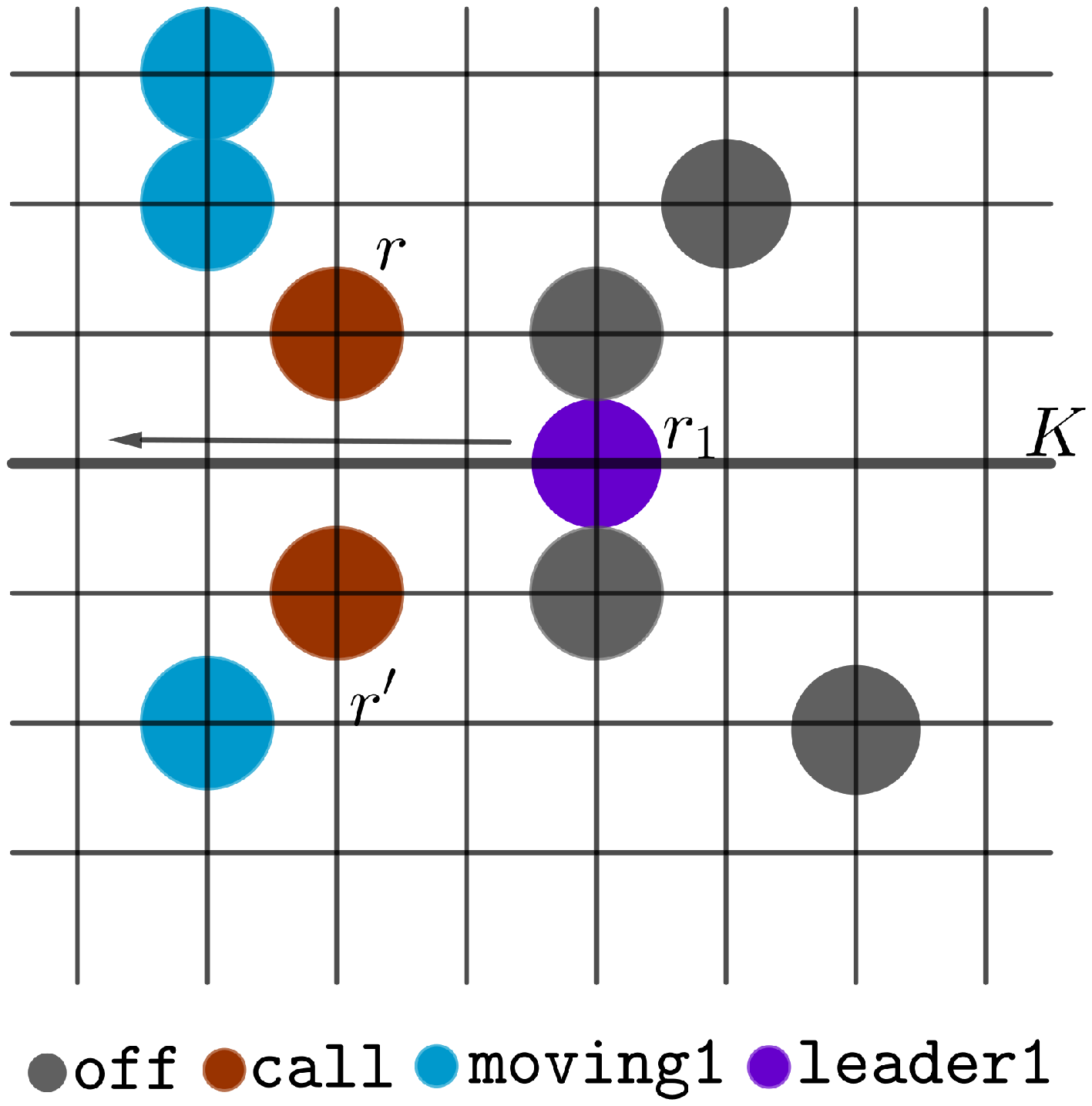}
     \caption{$r_1$ has colour \texttt{leader1} and has $\mathcal{L}_H(r_1) \cap H_L^O(r_1)$ empty.}\label{fig:Lemma5case1}
   \end{minipage}
   \hfill
   \begin {minipage}{0.45\textwidth}
    \includegraphics[height=6cm,width=7cm]{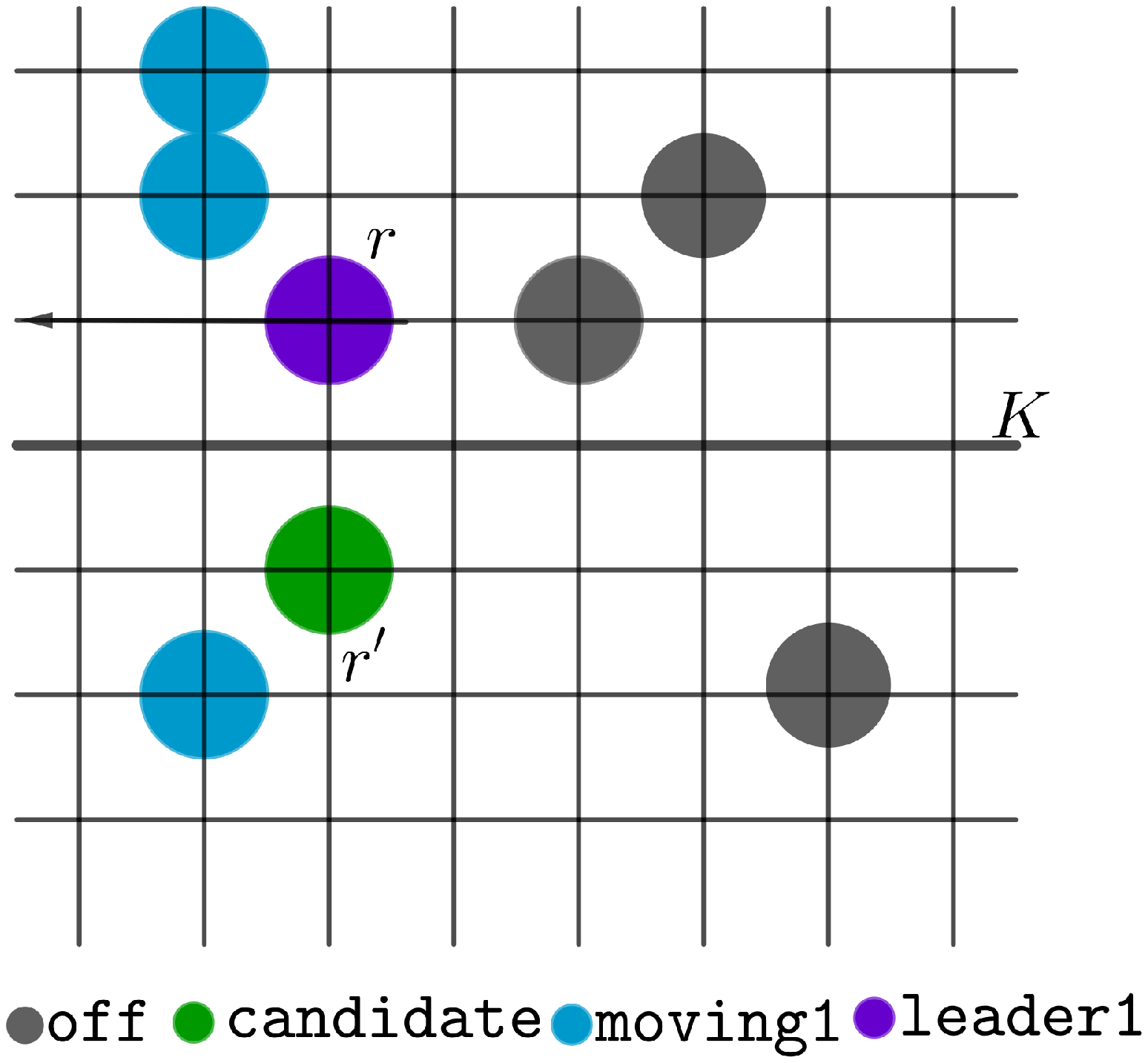}
     \caption{$r$ is on $\mathcal{L}_2$ with colour \texttt{leader1} and has $\mathcal{L}_H(r) \cap H_L^O(r)$  empty.}\label{Fig:Lemma5case2}
   \end{minipage}
\end{figure}

\textbf{Case-II:} Without loss of generality, let $r$ be the robot with colour \texttt{candidate} that changes the colour to \texttt{leader1}. Note that, either $r$ is on $\mathcal{L}_1$ or it is on $\mathcal{L}_2$. If $r$ is on $\mathcal{L}_1$ then it finds its left grid point is empty and moves to left and become singleton on $\mathcal{L}_1$. So, left grid point of $r$ is empty . Now if $r$ was on $\mathcal{L}_2$ (Figure \ref{Fig:Lemma5case2}). Then $\mathcal{L}_1 \cap \mathcal{L}_H(r)$ is empty as there are no robots between the line $\mathcal{L}_H(r)$ and $\mathcal{L}_H(r')$ on $\mathcal{L}_I(r)$. So, $r$ can move left and become singleton on $\mathcal{L}_1$. Hence $r$ always has its left grid point empty.

\end{proof}

\begin{lemma}
If  a robot $r$ with colour \texttt{call} does not see another robot with colour \texttt{leader1} on $\mathcal{R}_I(r)$, then there is a time $T$ when $\mathcal{L}_V(r)$ will always have a robot with colour \texttt{moving1} and two robots with colour \texttt{call} in $\mathbb{C}(T)$.
\end{lemma}
\begin{proof}
$r$ is a robot with colour \texttt{call} on $\mathcal{L}_V(r)$. This implies $\mathcal{R}_I(r)$ is symmetric with respect to $K$, where $K$ is known because there is another robot say $r'$ on $\mathcal{L}_V(r)$ with colour \texttt{call} or \texttt{candidate}. Note that if $r'$ has colour \texttt{candidate}, it changes the colour to \texttt{call} after a finite time. In this situation, if there is a robot say $r_1$ on $K \cap \mathcal{R}_I(r)$, then $r_1$  changes its colour to \texttt{leader1} from \texttt{off}. And also, $r$ sees $r_1$  on $\mathcal{R}_I(r)$. Since it is assumed that $r$ is not seeing any robot with colour \texttt{leader1} on $\mathcal{R}_I(r)$, it is evident that there is no robot on $K \cap \mathcal{R}_I(r)$. In this scenario, the robots on $\mathcal{R}_I(r)$ see that there are exactly two robots $r$ and $r'$ with colour \texttt{call} on left immediate vertical line. So, the robots on $\mathcal{R}_I(r)$, which are closest to $K$ change their colours to \texttt{moving1} upon activation and all the robots who can see a robot with colour \texttt{moving1} on their vertical line eventually change their colours to \texttt{moving1}. Observe that in this way, after a finite time there will be at least one robot on $\mathcal{R}_I(r)$ which has colour \texttt{moving1} and also will be terminal on $\mathcal{R}_I(r)$. Let $r_2$ be that robot. Now $r_2$ will move vertically in one fixed direction until at least one of $H_U^C(r) \cap \mathcal{L}_I(r)$ and $H_B^C(r) \cap \mathcal{L}_I(r)$ has no other robot and then it moves left to $\mathcal{L}_V(r)$ (Figure \ref{fig:Lemma6pic1}). Also, note that $r$ and $r'$ do not change their colours until $r_2$ reaches $\mathcal{L}_V(r)$. So, after a finite time say $T$, there will be a robot $r_2$ with colour \texttt{moving1} and two robots $r$ and $r'$ with colour \texttt{call} on $\mathcal{L}_V(r)$ in $\mathbb{C}(T)$ (Figure \ref{Fig:Lemma6pic2}).

\begin{figure}[!htb]\centering
   \begin{minipage}{0.45\textwidth}
    \includegraphics[height=6.5cm,width=7cm]{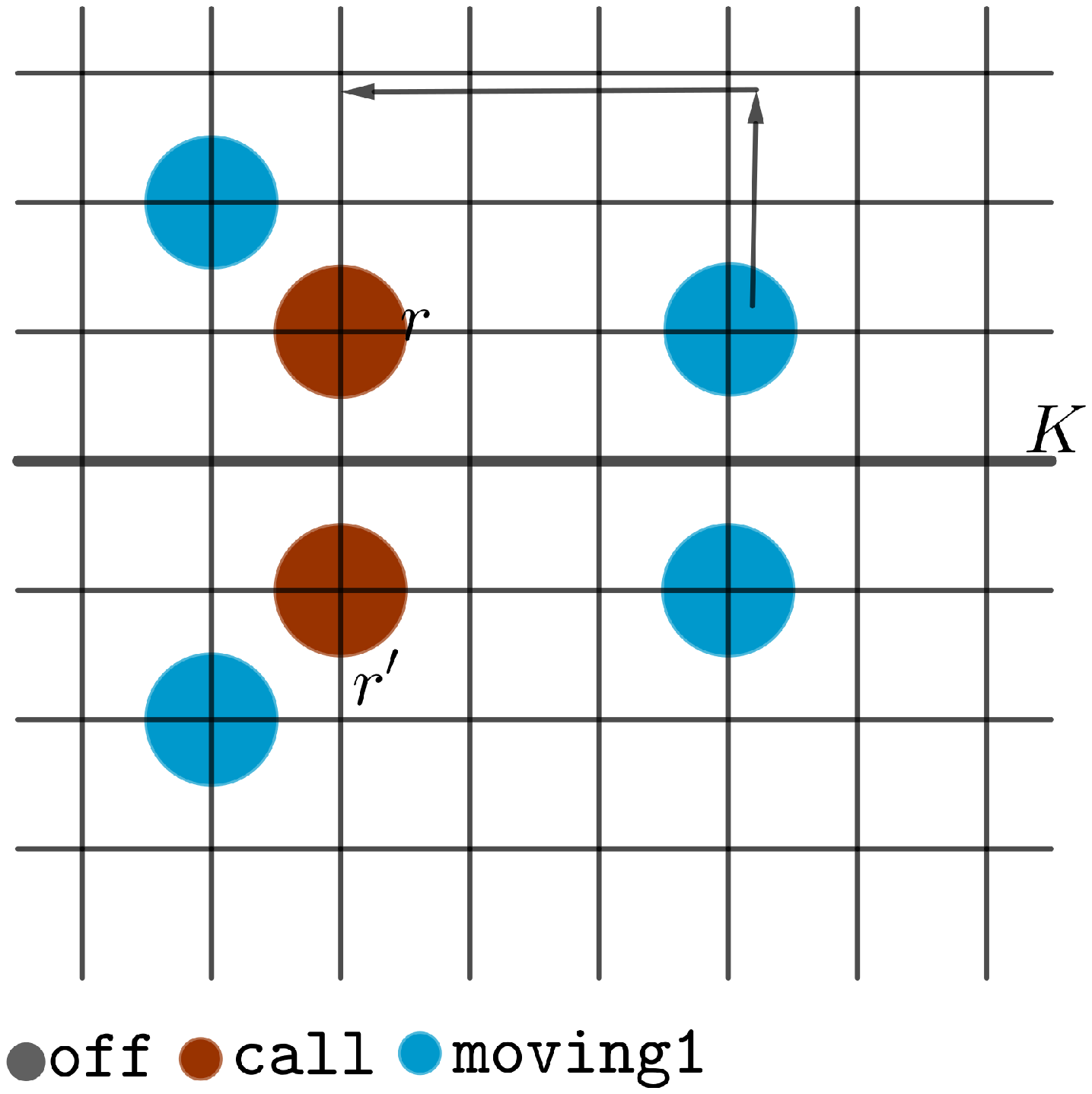}
     \caption{Terminal robot on $\mathcal{R}_I(r)$ see $r$ with colour \texttt{call} and move according to the path shown by the arrow.}\label{fig:Lemma6pic1}
   \end{minipage}
   \hfill
   \begin {minipage}{0.45\textwidth}
    \includegraphics[height=6cm,width=7cm]{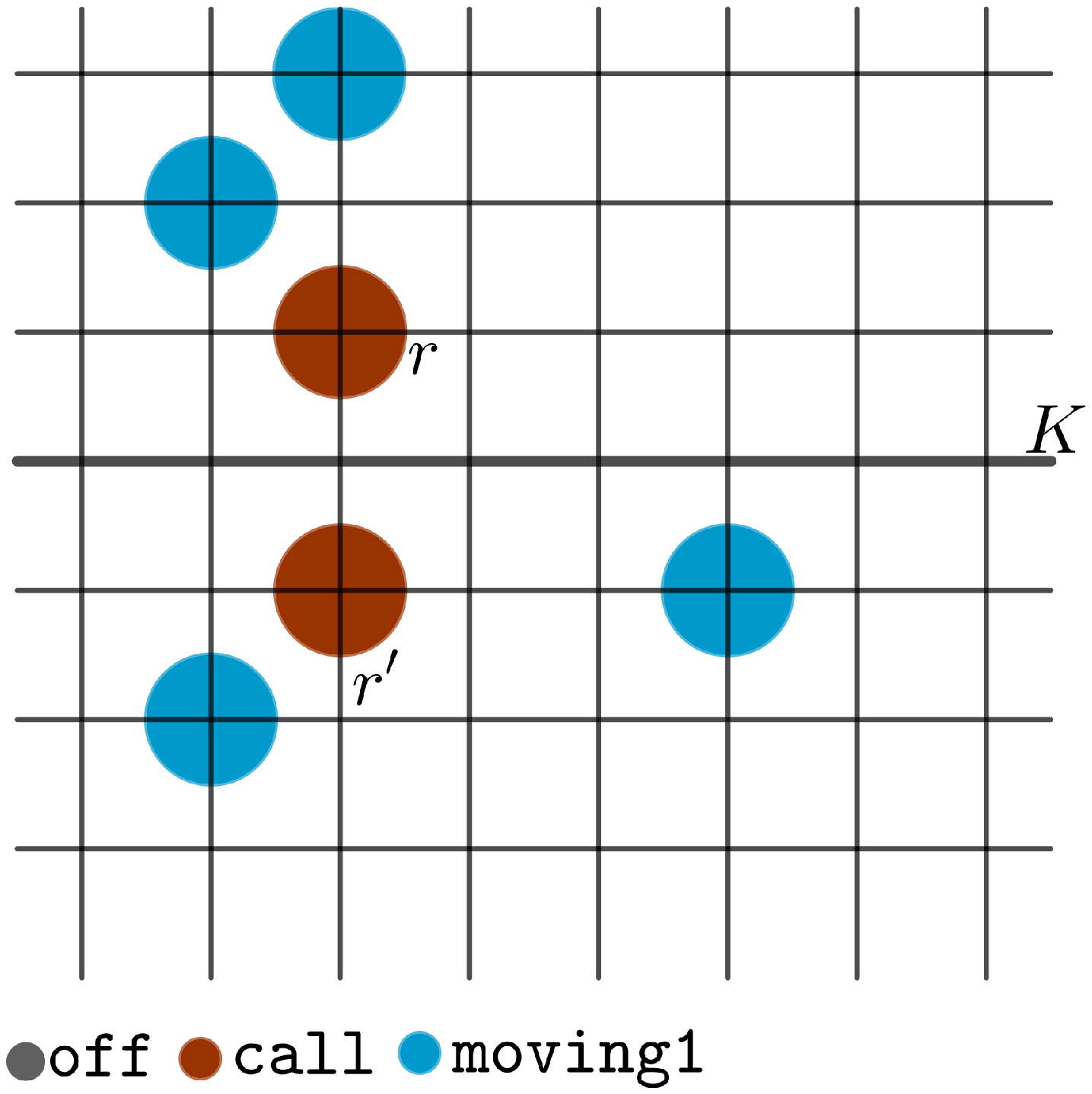}
     \caption{The robot with colour \texttt{moving1} reaches above $r$. In this moment, $\mathcal{L}_V(r)$ has two robots $r$ and $r'$ with colour \texttt{call} and one robot with colour \texttt{moving1}.}\label{Fig:Lemma6pic2}
   \end{minipage}
\end{figure}
\end{proof}
\begin{lemma}
During movement of robots with colour \texttt{moving1} in \textit{Phase 1}, no collision occurs.
\end{lemma}
\begin{proof}
A robot $r$ with colour \texttt{moving1} can have two type of moves, horizontal to the left and vertical. $r$ moves vertically on $\mathcal{L}_V(r)$ only when it sees at least one robot with colour \texttt{call} on $\mathcal{L}_I(r)$. Note that during vertical movement of $r$, no other robot on $\mathcal{L}_V(r)$ moves vertically in the same direction as $r$. This is because if another robot say $r'$ moves on $\mathcal{L}_V(r)$, it must be terminal on $\mathcal{L}_V(r)$ and has colour \texttt{moving1}. But since $r$ is already on $\mathcal{L}_V(r)$, $r'$ moves opposite of $r$. So, as long as $r$ moves vertically on $\mathcal{L}_V(r)$ no collision occurs. Note that $r$ moves vertically in such a way such that at least one of $H_U^C(r) \cap \mathcal{L}_I(r)$ or $H_B^C(r) \cap \mathcal{L}_I(r)$ has no other robot and then it moves left towards $\mathcal{L}_I(r)$ (i.e the same vertical line were the robot with colour \texttt{call} is located).  Now let there is a non-terminal robot $r_1$ which is nearest to $r$ and  below $r$ on $\mathcal{L}_V(r)$ with colour \texttt{moving1}. Now observe that $r_1$ only moves when $r$ reaches the vertical line of the robot with colour \texttt{call}. In this scenario, $r_1$ moves vertically in such a way such that it has either  $H_U^C(r_1) \cap \mathcal{L}_I(r_1)$ or $H_B^C(r_1) \cap \mathcal{L}_I(r_1)$ has no other robot and then moves left to the empty grid point. So, during horizontal or vertical movement of robots with colour \texttt{moving1}, no collision occurs. Hence the result.
\end{proof}

\begin{lemma}
If at time $T$, two robots $r$ and $r'$ have colour \texttt{call}  on the same vertical line  and there is no robot on $K \cap \mathcal{R}_I(r)$, then there exist $T'>T$ such that both $r$ and $r'$ are with colour  \texttt{reached}  at $\mathbb{C}(T')$.
\end{lemma}
\begin{proof}
Let $r$ and $r'$ be two robots with colour \texttt{call} at time $T$ on same vertical line $\mathcal{L}_V(r)$ (i.e. $\mathcal{L}_V(r')$). Then $\mathcal{R}_I(r)$ must be symmetric with respect to $K$. Also, there is no robot on $K \cap \mathcal{R}_I(r)$. So, no robot with colour \texttt{off} on $\mathcal{R}_I(r)$ changes its colour to \texttt{leader1}. Now in this scenario, the robots which are closest to $K$ on $\mathcal{R}_I(r)$ change their colours to \texttt{moving1}. Note that a robot with colour \texttt{off} also can change its colour to \texttt{moving1} if it sees another robot with colour \texttt{moving1} on the same vertical line. Also, no robot with colour \texttt{moving1} moves unless it is terminal on the same vertical line. Hence we can say that at least all robots of above or below $K$ on $\mathcal{R}_I(r)$ change their colours to \texttt{moving1}. Now by Algorithm \ref{leader_selection}, the terminal robots  with colour \texttt{moving1} move to $\mathcal{L}_V(r)$. Then next robot becomes terminal and do the same. So, after a finite time say $T_1 > T$, all robots with colour \texttt{moving1} on $\mathcal{R}_I(r)$ move to $\mathcal{L}_V(r)$. In this moment, all robots of $\mathcal{R}_I(r)$ have colour \texttt{off}. Note that in this scenario, at least one of $r$ or $r'$ must see a robot with colour \texttt{moving1} on the same vertical line upon activation. Without loss of generality, let $r$ sees a robot with light \texttt{moving1} on $\mathcal{L}_V(r)$ and all robots on $\mathcal{R}_I(r)$ have colour \texttt{off} (Figure \ref{fig:Lemma8pic1}). Then $r$ changes its colour to \texttt{reached} at time say $T_2 \ge T_1\ge T$ (Figure \ref{Fig:Lemma8pic2}). Now when $r'$ activates, it sees $r$ with colour \texttt{reached} on $\mathcal{L}_V(r')$ and changes its colour to \texttt{reached} at a time $T_3 \ge T_2$ (here $T' = T_3$) (Figure \ref{fig:Lemma8pic3}). Now it may be possible due to asynchronous environment that after $r$ changes its colour to \texttt{reached} at time $T_2$,  a robot say $r_1$, on $\mathcal{R}_I(r)$ changes its colour to \texttt{moving1} . Then $r'$ will not change its colour to \texttt{reached} now, even after seeing $r$ with colour \texttt{reached} as all robots on $\mathcal{R}_I(r')$ now do not have colour \texttt{off}. 
Now when $r_1$ wakes again at a time say $T_4 (\ge T_2)$, it sees $r$ with colour \texttt{reached} on $\mathcal{L}_I(r)$ and changes its colour to \texttt{off}. Now when $r'$ wakes again at some time $T' \ge T_4 \ge T_2 \ge T_1 > T$, it changes its colour to \texttt{reached}. Note that $r$ does not change its colour from \texttt{reached} to \texttt{candidate} before $r'$ wakes and changes its colour to \texttt{reached}  as it will not see any other robot with colour \texttt{reached} or \texttt{candidate} on $\mathcal{L}_V(r)$ before $r'$ wakes. So, we can conclude that $\exists$  $T' > T$ when both $r$ and $r'$ have colour \texttt{reached}.

\begin{figure}[!htb]\centering
   \begin{minipage}{0.45\textwidth}
    \includegraphics[height=6cm,width=7cm]{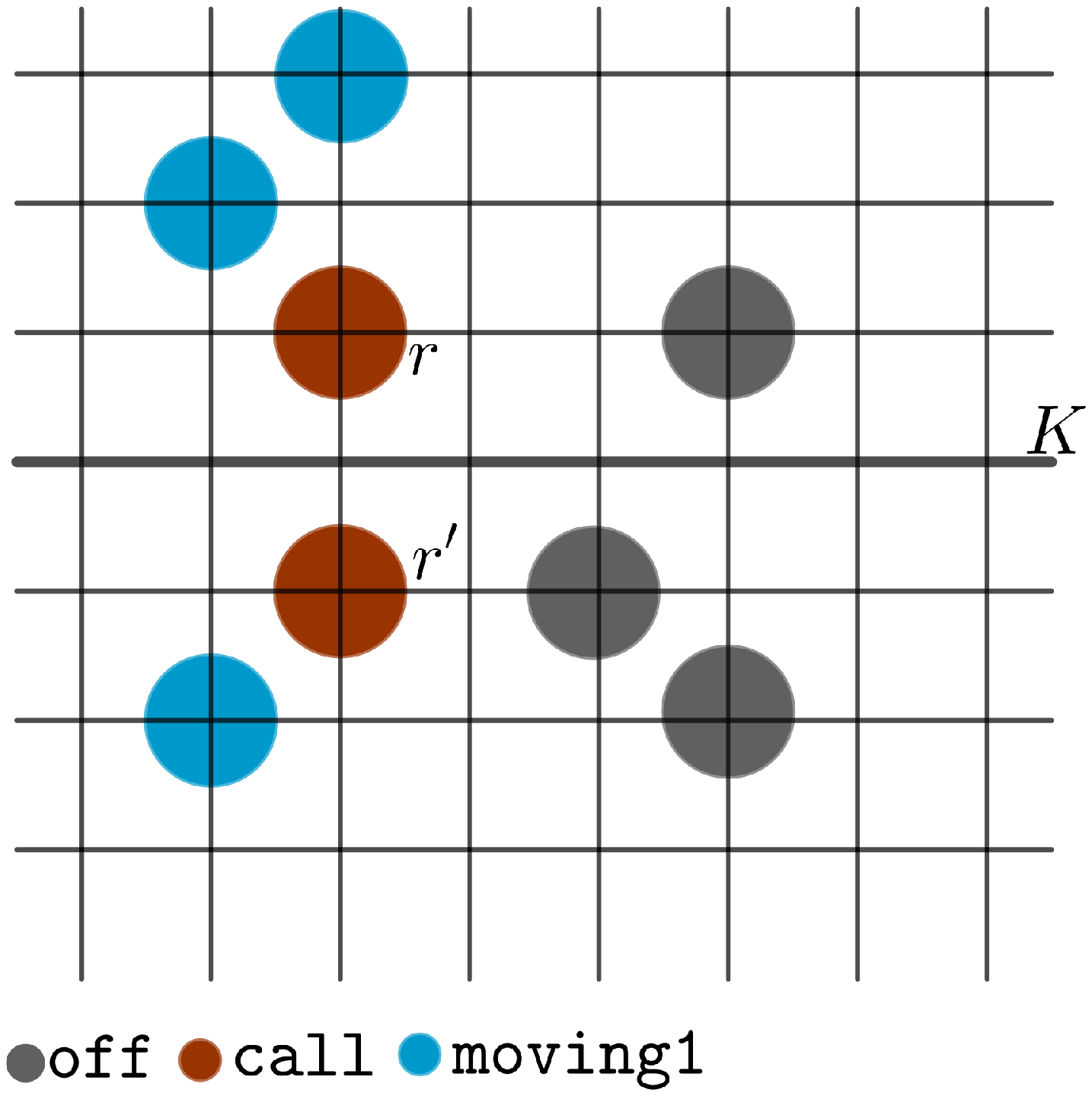}
     \caption{$r$ sees a robot with colour \texttt{moving1} on $\mathcal{L}_V(r)$ and sees all robots on $\mathcal{R}_I(r)$ with colour \texttt{off}. $r'$ does not see any robot with colour \texttt{moving1} or \texttt{reached} on $\mathcal{L}_V(r')$.}\label{fig:Lemma8pic1}
   \end{minipage}
   \hfill
   \begin {minipage}{0.45\textwidth}
    \includegraphics[height=6cm,width=7cm]{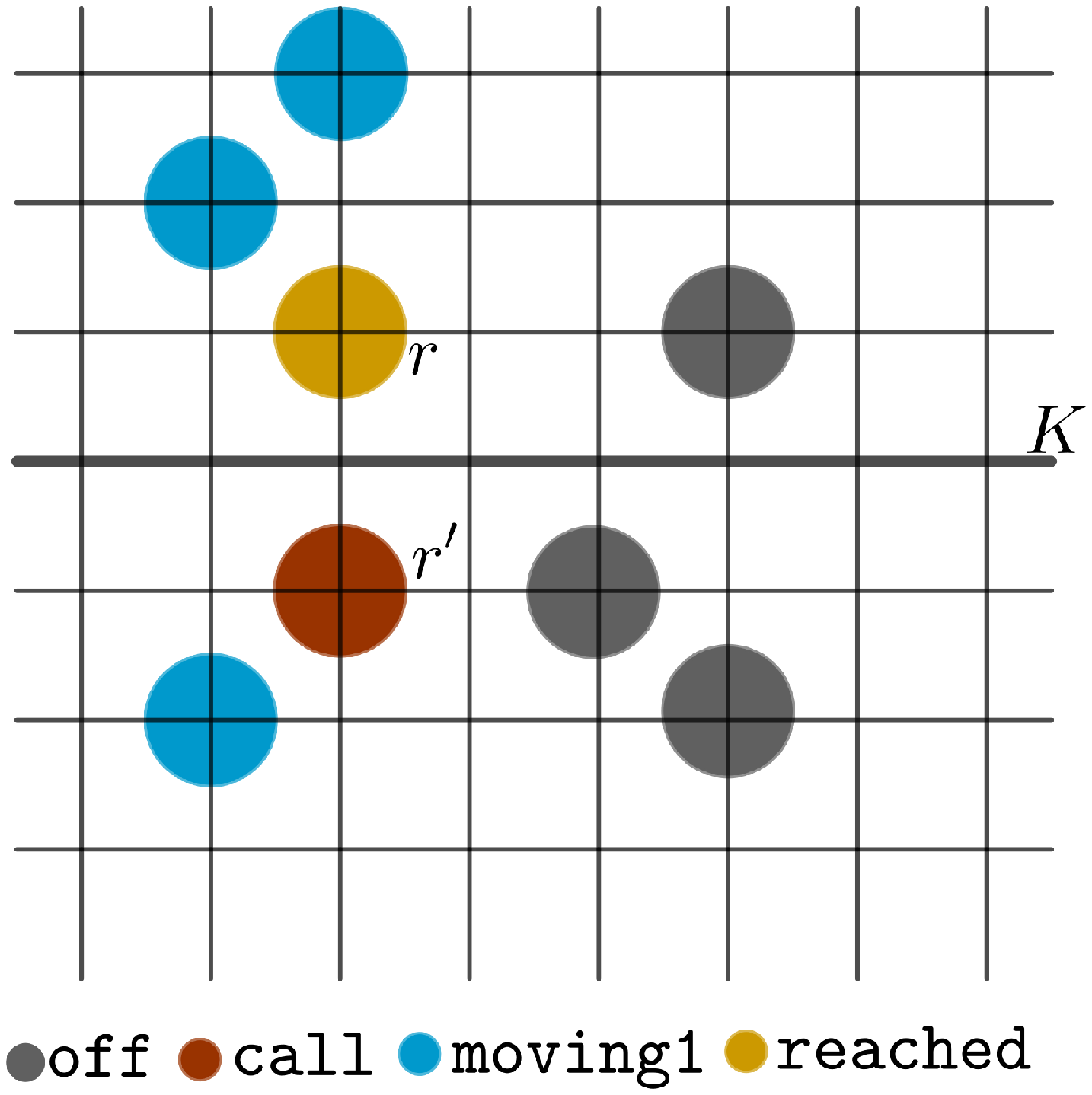}
     \caption{$r$ changes its colour to \texttt{reached}.}\label{Fig:Lemma8pic2}
   \end{minipage}
\end{figure}
\begin{figure}
    \centering
    \includegraphics[height=6cm,width=7cm]{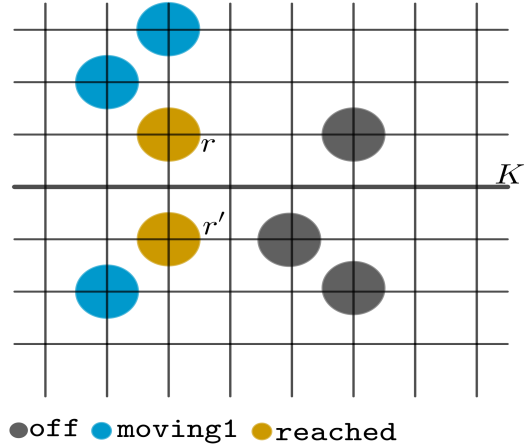}
    \caption{Now $r'$ sees $r$ with colour \texttt{reached} on $\mathcal{L}_V(r')$ and all robots on $\mathcal{R}_I(r')$ with colour \texttt{off}. So, $r'$ changes its colour to \texttt{reached}.}
    \label{fig:Lemma8pic3}
\end{figure}

\end{proof}

\begin{lemma}
If a robot $r$ changed its colour to \texttt{reached} at some time $T >0$, then $\exists$ $T' \ge T$ such that all robot in $\mathcal{R}_I(r)$ in $\mathbb{C}(T')$ have colour \texttt{off}.
\end{lemma}
\begin{proof}
During the look phase, $r$ must have seen robots on $\mathcal{R}_I(r)$ have colour \texttt{off}. Now if no robot on $\mathcal{R}_I(r)$ change its colour to \texttt{moving1} in between the completion of look phase of $r$ and time $T$, then $T' = T$. Now if a robot, say $r_1$ changes its colour to \texttt{moving1} in between completion of look phase of $r$ and time $T$, then there exists $T' > T$ when $r_1$ sees a robot with colour \texttt{reached} on $\mathcal{L}_I(r_1)$ and so changes its colour to \texttt{off}. Note that before $r_1$ changes its colour to \texttt{off}, $r$ does not change its colour as it sees $r_1$ with colour \texttt{moving1} on $\mathcal{R}_I(r)$. So, we can conclude $\exists$ $T' \ge T$ such that all robots on $\mathcal{R}_I(r)$ have colour \texttt{off} in $\mathbb{C}(T')$. 
\end{proof}

\begin{lemma}
\label{flemma10}
If at time $T$, a robot changes its colour to \texttt{leader1} from \texttt{off}, then $\mathbb{C}(T')$ has no robot with colour \texttt{candidate} or \texttt{terminal1}, where $T' \ge T$.
\end{lemma}
\begin{proof}
If $r$ changes its colour to \texttt{leader1} from \texttt{off} at some time $T$, then it must have seen  exactly two robots say $r_1$ and $r_2$ with colour \texttt{call} on $\mathcal{L}_I(r)$ at a time $T_1$ where $T_1 < T$ (Figure \ref{fig:Lemma10pic1}). Note that a robot can only have colour \texttt{call} at some time $T_2$ if it had colour \texttt{candidate} at some time $T_3 < T_2$. Also, a robot can change its colour to \texttt{candidate} from \texttt{terminal1} only if it sees there is no other robot on its left open half. Also, a robot with colour \texttt{off} changes to colour \texttt{terminal1} only if its left open half empty, there is no robot with colour \texttt{leader1} on $\mathcal{R}_I(r_1)$ or on $\mathcal{L}_V(r_1)$ and  it is terminal on $\mathcal{L}_V(r)$. Since during the whole execution of \textit{Phase 1}, no other robot having colour \texttt{off} except $r_1$ and $r_2$ can see its left open half empty and find themselves to be terminal, no other robot except $r_1$ and $r_2$ can change their colours to \texttt{terminal1}. Now upon activation again  at any time $T_4 > T$ , both $r_1$ and $r_2$ sees $r$ on $\mathcal{R}_I(r_1)$ with colour \texttt{leader1} and change their colours to \texttt{off} (Figure \ref{Fig:Lemma10pic2}). Observe that after time $T_4$,  $r_1$ and $r_2$ can never change their colour to \texttt{terminal1} and hence to \texttt{candidate} as they will see $r$ with colour \texttt{leader1} on $\mathcal{R}_I(r_1)$ or on $\mathcal{L}_V(r_1)$ or $r_1$ and $r_2$ would have its left  open half non-empty. So, we can conclude the lemma.

\begin{figure}[!htb]\centering
   \begin{minipage}{0.45\textwidth}
    \includegraphics[height=6cm,width=6.5cm]{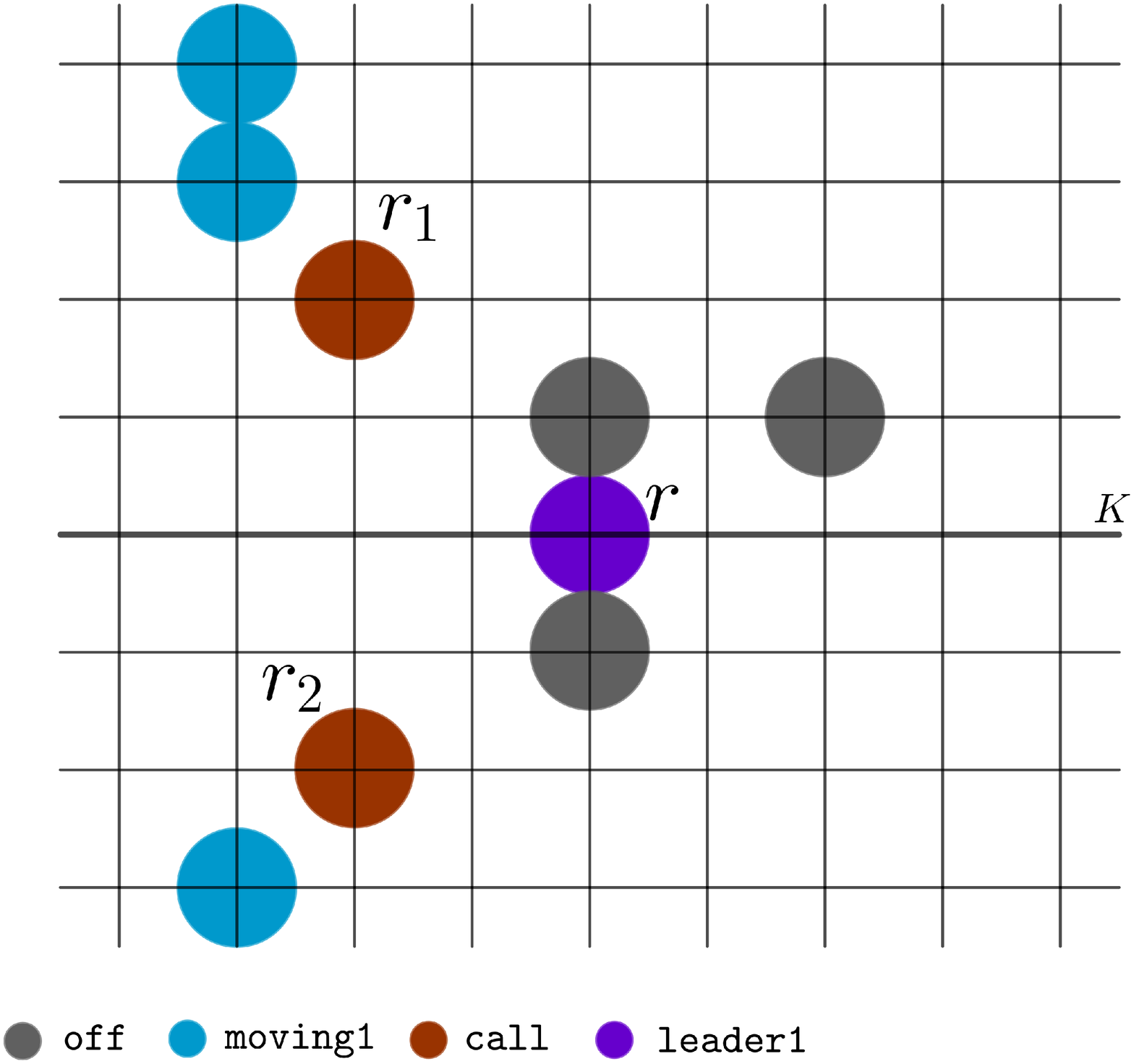}
     \caption{$r_1$ and $r_2$ with colour \texttt{call} both see $r$ with colour \texttt{leader1} on $\mathcal{R}_I(r_1) \cap K$.}\label{fig:Lemma10pic1}
   \end{minipage}
   \hfill
   \begin {minipage}{0.45\textwidth}
    \includegraphics[height=6cm,width=6.5cm]{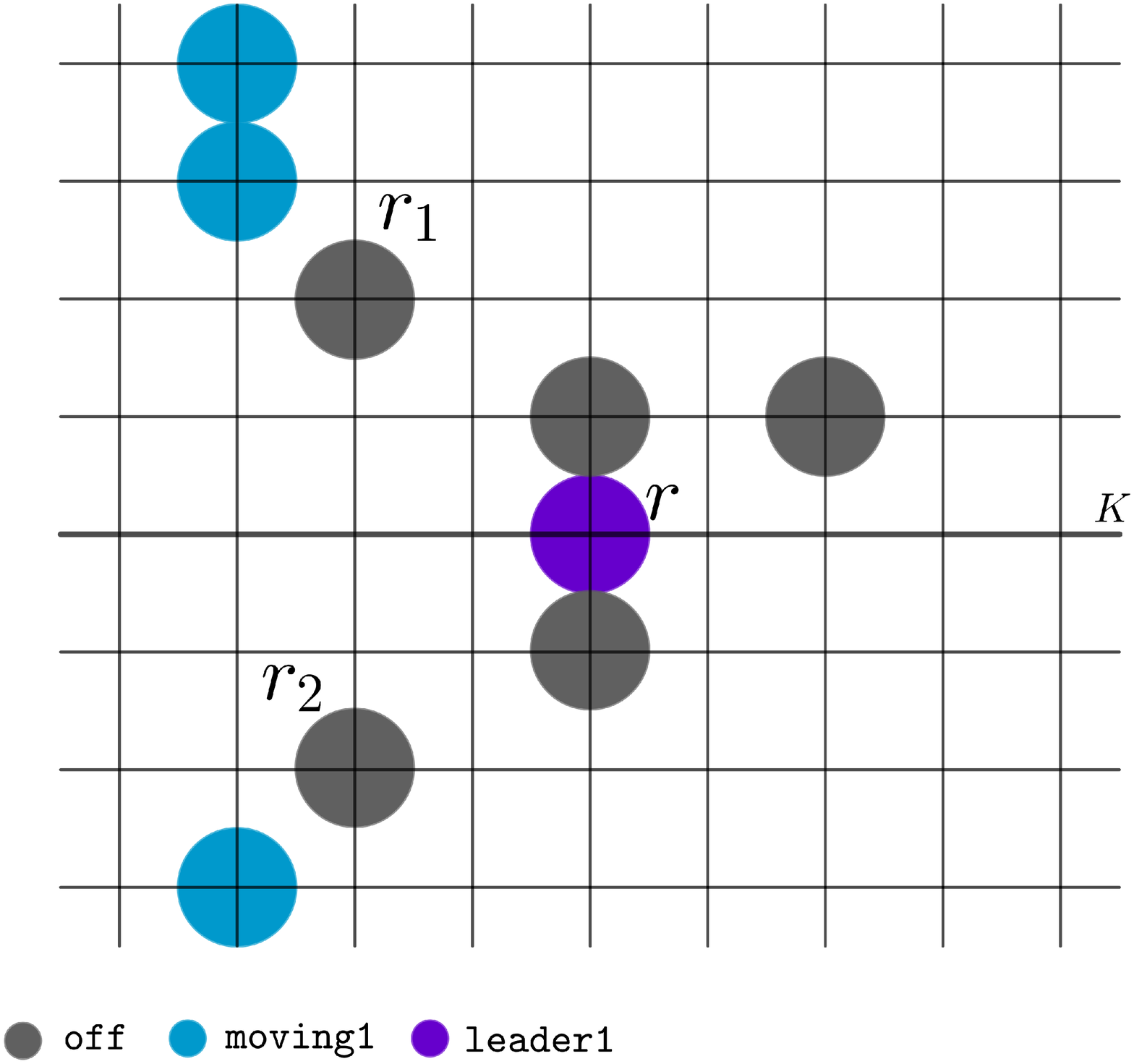}
     \caption{Both $r_1$ and $r_2$ change their colours to \texttt{off} after seeing $r$.}\label{Fig:Lemma10pic2}
   \end{minipage}
\end{figure}
\end{proof}

\begin{lemma}
\label{flemma11}
If at a time $T$, a robot $r$ changed its colour to \texttt{leader1}, then there will be no robot with colour \texttt{reached} in $\mathbb{C}(T')$, where $T' \ge T$.
\end{lemma}
\begin{proof}
Note that a robot can only change its colour to \texttt{reached} at a time $T$ if $\exists$ $T_1 < T$ such that the robot had colour \texttt{candidate} in $\mathbb{C}(T_1)$. Now if $r$ changed its colour to \texttt{leader1} from \texttt{candidate}, then even if there is another robot say $r'$ with colour \texttt{candidate} on $\mathcal{L}_V(r)$ (Figure \ref{fig:Lemma11pic1}), $r'$ will change its colour to \texttt{off} upon first activation at a time $T_2 > T$. So, the configuration now has no robot with colour \texttt{candidate} or \texttt{reached} (as both $r$ and $r'$ with colour \texttt{candidate} who could have changed their colour to \texttt{reached} changed it to \texttt{leader1} and \texttt{off}) (Figure \ref{Fig:Lemma11pic2}). Also, note that during the period between $T$ and $T_2$, the configuration does not have colour \texttt{reached} as in this time $r$ has colour \texttt{leader1} and $r'$ has colour \texttt{candidate}. Also, no other robot with colour \texttt{off} will ever change its colour to \texttt{candidate} after time $T_2$ as a robot say $r_1$ with colour \texttt{off} or \texttt{terminal1} either sees $r$ with colour \texttt{leader1} on $\mathcal{L}_V(r_1)$ or on $\mathcal{R}_I(r_1)$ or it has its left open half non-empty. And since a robot can only change its colour to \texttt{reached} when  it had colour \texttt{candidate} before, there will be no robot with colour \texttt{reached} in $\mathbb{C}(T')$, where ($T' \ge T$).

\begin{figure}[!htb]\centering
   \begin{minipage}{0.45\textwidth}
    \includegraphics[height=6cm,width=6cm]{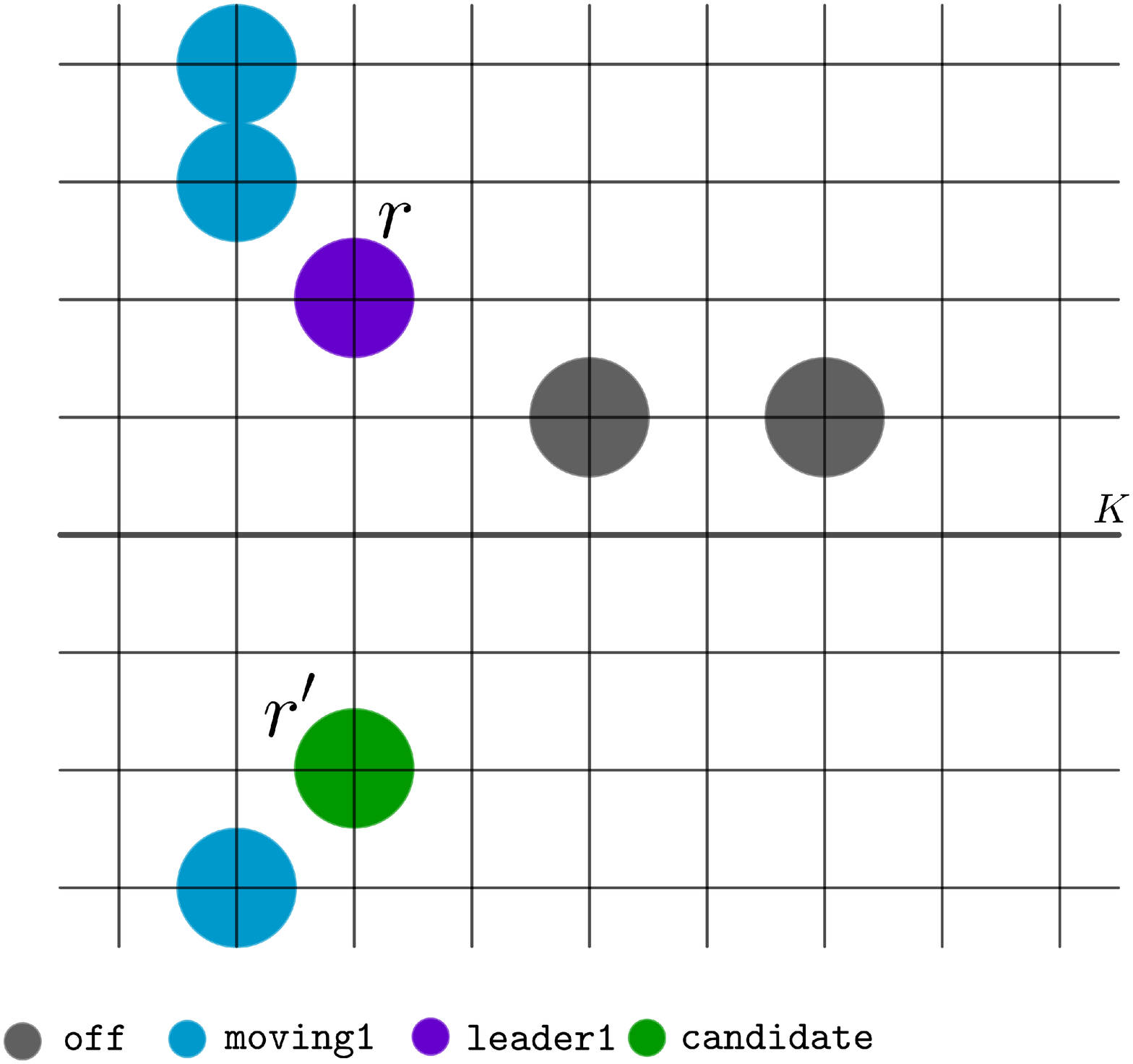}
     \caption{$r'$ with colour \texttt{candidate} sees $r$ with colour \texttt{leader1} on $\mathcal{L}_V(r')$.}\label{fig:Lemma11pic1}
   \end{minipage}
   \hfill
   \begin {minipage}{0.45\textwidth}
    \includegraphics[height=6cm,width=6cm]{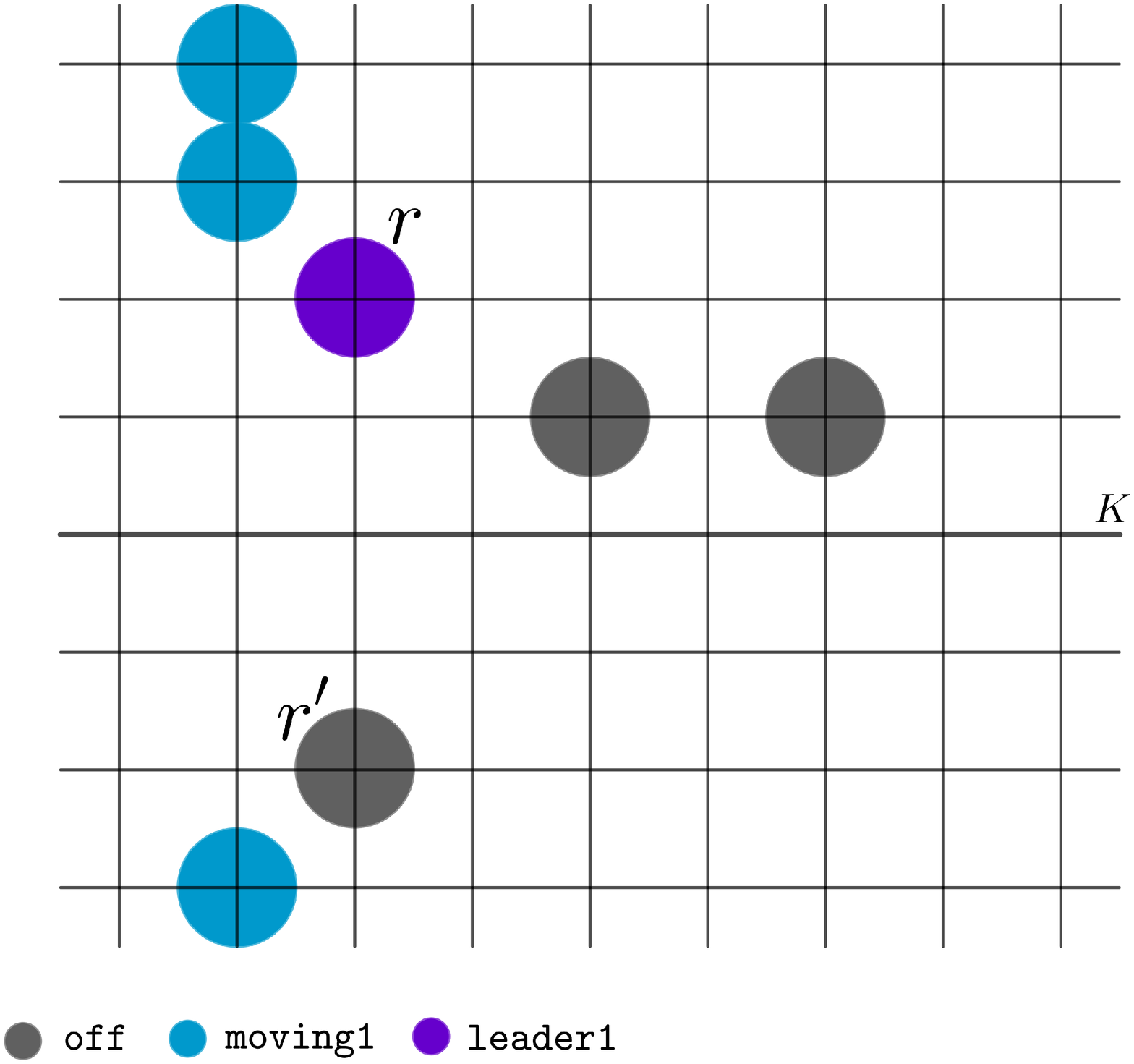}
     \caption{$r'$ changes its colour to \texttt{off}. Now, no other robot changes colour to \texttt{terminal1} and hence to \texttt{candidate} and hence to \texttt{reached}.}\label{Fig:Lemma11pic2}
   \end{minipage}
\end{figure}

Now, if $r$ has changed its colour to \texttt{leader1} from \texttt{off} at time $T$, then $r$ must have seen two robots say, $r_1$ and $r_2$ on $\mathcal{L}_I(r)$ with colour \texttt{call} at some time $T_1 < T$. Now upon activation after time $T$, both $r_1$ and $r_2$ see $r$ on $\mathcal{R}_I(r_1)$ and turn their colours to \texttt{off}. Now for $r_1$ and $r_2$ to ever have the colour \texttt{reached} again must have colour \texttt{candidate} first. But by Lemma \ref{flemma10}, after $r$ changes its colour to \texttt{leader1}, the configuration can never have a robot with colour \texttt{candidate}. Hence, $\mathbb{C}(T')$ ($T' \ge T$) has no robot with colour \texttt{reached} if $r$ changed its colour to \texttt{leader1} at time $T$.
\end{proof}

\begin{lemma}
At any time $T$, there can be at most one robot with colour \texttt{leader1} and at most one robot with colour \texttt{leader} in the configuration.
\end{lemma}
\begin{proof}
 Note that by Lemma \ref{flemma2}, a robot can change its colour to \texttt{leader1} only from the colour \texttt{off} or \texttt{candidate}. Let us consider the following cases:

\textbf{Case-I:} Consider the case where a robot changes its colour to \texttt{leader1} from the colour \texttt{candidate}. Now from Theorem \ref{fth1}, for any initial configuration $\mathbb{C}(0)$, there exist a time $T$ such that $\mathbb{C}(T)$ has either one robot with colour \texttt{leader1}  who has changed its colour to \texttt{leader1} from \texttt{candidate} or two robots with colour \texttt{candidate} on the same vertical line.

\textbf{Case-I(a):} Let $r$ is a robot with colour \texttt{leader1} in $\mathbb{C}(T)$ who has changed its colour from \texttt{candidate}. We claim that in $\mathbb{C}(T')$ where $T' \ge T$, there is no other robot who changes its colour to \texttt{leader1}. For this, we first show that  no other robot with colour \texttt{terminal1} ever change their colour to \texttt{candidate}. This is because no other robot with colour \texttt{terminal1} will find its left open half empty (as the robot with colour \texttt{leader1} is there) (Figure \ref{fig:Lemma12case1a}). So $\mathbb{C}(T')$, where $T' \ge T$, will not have any robot with colour \texttt{candidate} who can change further to \texttt{leader1}. Also observe that after at $T' \ge T$, no other robot with colour \texttt{off} changes its colour to \texttt{leader1} as they will not see any robot with colour \texttt{call} on their left immediate occupied vertical line. This is also for the reason that $\mathbb{C}(T')$ where $T' \ge T$ will not have any other robot with colour \texttt{candidate} who can change its colour to \texttt{call} further. So, there will be exactly one robot with colour \texttt{leader1} which eventually changes its colour to \texttt{leader}.

\begin{figure}[ht]
    \centering
    \includegraphics[height=6cm,width=7cm]{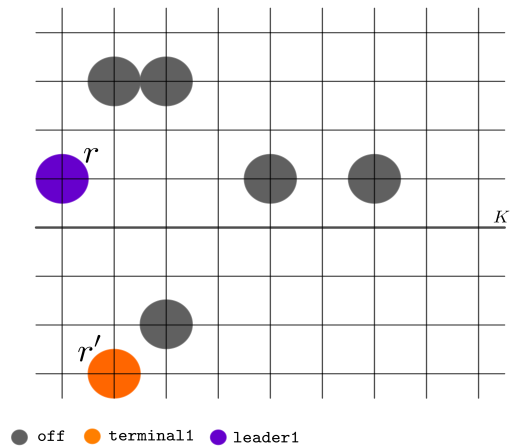}
    \caption{$r'$ with colour \texttt{terminal1} sees $r$ with colour \texttt{leader1} on $\mathcal{L}_I(r')$. So, $r'$ does not change its colour as it does not have its left open half empty.}
    \label{fig:Lemma12case1a}
\end{figure}

\textbf{Case-I(b):} Let us now assume the case where there are two robots $r_1$ and $r_2$ with colour \texttt{candidate} both on the same vertical line $\mathcal{L}_V(r_1)$(i.e. $\mathcal{L}_V(r_2)$). In this case, we will first show that both $r_1$ and $r_2$ can not change their colour to \texttt{leader1}. Then we will show if one of $r_1$ or $r_2$ changes its colours to \texttt{leader1}, then no other robot with colour \texttt{off} changes its colour to \texttt{leader1}.

\begin{figure}[!htb]\centering
   \begin{minipage}{0.45\textwidth}
    \includegraphics[height=6cm,width=6cm]{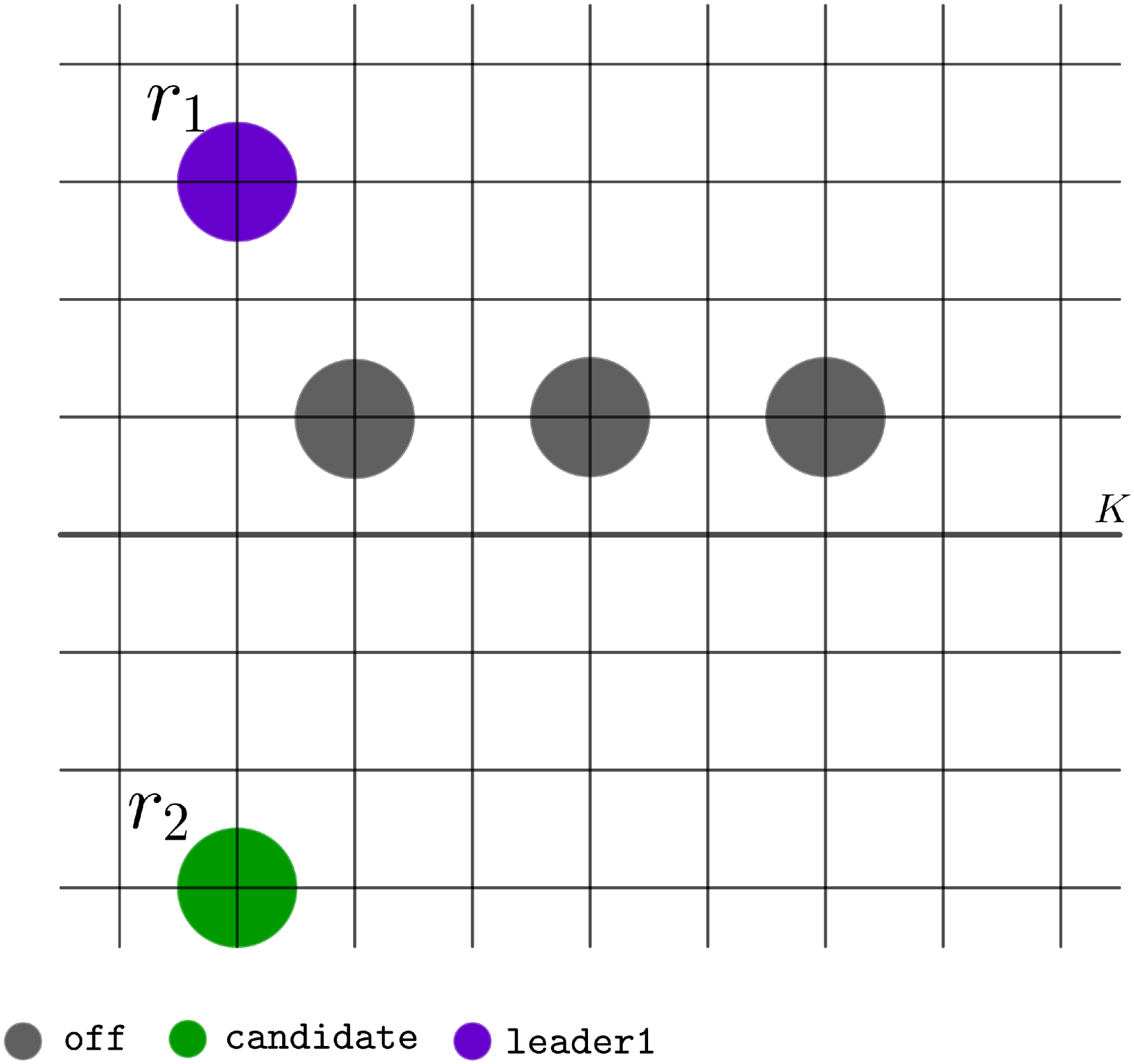}
     \caption{$r_2$ with colour \texttt{candidate} sees $r$ with colour \texttt{leader1} on $\mathcal{L}_V(r_2)$. No robot with colour \texttt{call} in the configuration.}\label{fig:Lemma12case1bpic1}
   \end{minipage}
   \hfill
   \begin {minipage}{0.45\textwidth}
    \includegraphics[height=6cm,width=6cm]{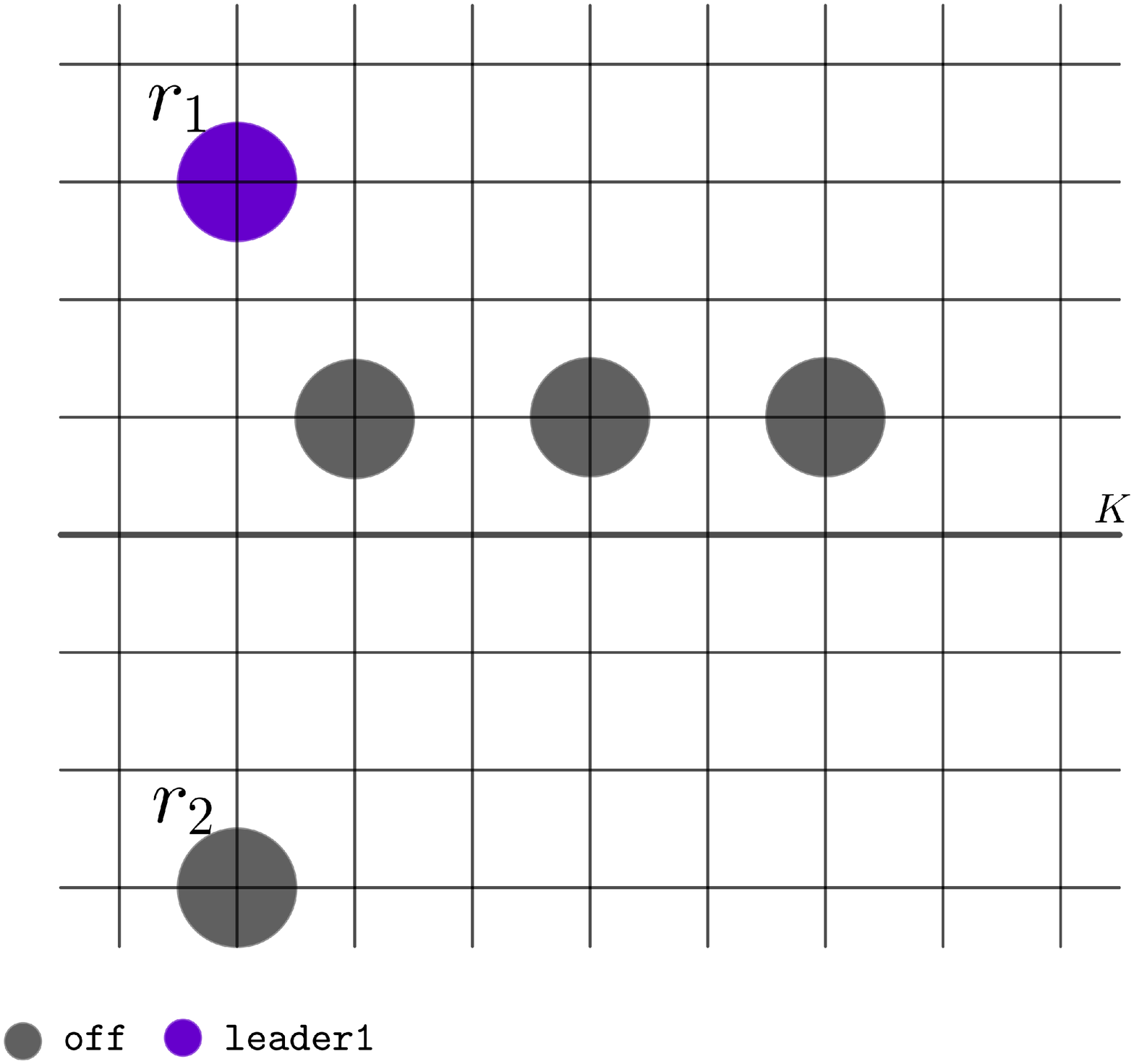}
     \caption{$r_2$ changes its colour to \texttt{off}. No robot with colour \texttt{call} in the configuration and no other robot with colour \texttt{off} changes colour to \texttt{terminal1} and hence to \texttt{candidate} and hence to \texttt{leader1} or \texttt{call} . }\label{Fig:Lemma12case1bpic2}
   \end{minipage}
\end{figure}

In this case, $r_1$ and $r_2$ check the symmetry of the line $\mathcal{R}_I(r_1)$(i.e $\mathcal{R}_I(r_2)$). If $\mathcal{R}_I(r)$ is asymmetric, then the robot ($r_1$ or, $r_2$) whichever is  on the dominant half changes its colour to \texttt{leader1}. Without loss of generality. let at some time $T_1$, $r_1$  changes its colour to \texttt{leader1} from \texttt{candidate}. Then $r_2$ must have colour \texttt{candidate} in $\mathbb{C}(T_1)$ (Figure \ref{fig:Lemma12case1bpic1}). Now when $r_2$ wakes again at a time say $T_2 > T_1$, it sees $r_1$ with colour \texttt{leader1} on $\mathcal{L}_V(r_2)$ and changes its colour to \texttt{off} (Figure \ref{Fig:Lemma12case1bpic2}). Note that between time $T_1$ and $T_2$ even if $r_1$ awakes again, it does not move as it sees $r_2$ with colour \texttt{candidate} on $\mathcal{L}_V(r_1)$. Now even if $r_2$ is terminal with colour \texttt{off} and has left open half empty, it would not change its colour to \texttt{terminal1} and then to \texttt{candidate} again as it sees $r_1$ with colour \texttt{leader1} on the same vertical line. So, between two robots with colour \texttt{candidate} only one can change its colour to \texttt{leader1}.


Now a robot say $r$ with colour \texttt{off} can never change its colour to \texttt{leader1} as it would not see exactly two robots with colour \texttt{call} on $\mathcal{L}_I(r)$. This is because a robot can only change its colour to \texttt{call} from colour \texttt{candidate} and no other robot with colour \texttt{off} will ever change its colour to \texttt{terminal1} and then to \texttt{candidate}  as in this case even if a robot with colour \texttt{off} has its left open half empty and is terminal on its vertical line, it will see $r_1$ with colour \texttt{leader1} on the same vertical line (Figure \ref{Fig:Lemma12case1bpic2}). So, it would not change its colour to \texttt{terminal1}. So, if a robot changes its colour to \texttt{leader1} from \texttt{candidate}, then no other robot will change its colour to \texttt{leader1}.

\textbf{Case-II:} Next we show that if a robot say $r$ has changed its colour to \texttt{leader1} from \texttt{off}, then no other robot with colour \texttt{candidate} or \texttt{off} ever changes its colour to \texttt{leader1}.

 Let $r$ changed its colour to \texttt{leader1} from colour \texttt{off} at some time $T_2$. This implies $r$ must have seen exactly two robots say $r_1$ and $r_2$ with colour \texttt{call} on $\mathcal{L}_I(r)$ and $r$ is on $K \cap \mathcal{L}_V(r)$. Also $\mathbb{C}(T_2)$ has no robot with colour \texttt{candidate} (Figure \ref{fig:Lemma12case2pic1}). We will now show that no robot will ever change its colour to \texttt{candidate} again. Now when $r_1$ and $r_2$ wake again (lets say at time $T_2' > T_2$), it sees $r$ with colour \texttt{leader1} on $\mathcal{R}_I(r_1)$ and so change their colours to \texttt{off} (Figure \ref{Fig:Lemma12case2pic2}). Observe that, in this scenario there is no robot with colour \texttt{reached} and \texttt{terminal1} in the configuration $\mathbb{C}(T_2')$ and no robot with colour \texttt{off} will ever change its colour to \texttt{terminal1} and then to \texttt{candidate} eventually. This is because even if a robot with colour \texttt{off} finds its left open half empty and it is terminal on its vertical line, it sees $r$ with colour \texttt{leader1} on its right immediate vertical line or on its own vertical line. So after time $T_2'$, the configuration has no robot with colour \texttt{candidate}, so no robot with colour \texttt{call} or \texttt{reached}. Thus if $r$ changed its colour to \texttt{leader1}, no other robot can change its colour to \texttt{leader1} from \texttt{candidate}.

\begin{figure}[!htb]\centering
   \begin{minipage}{0.45\textwidth}
    \includegraphics[height=6cm,width=6cm]{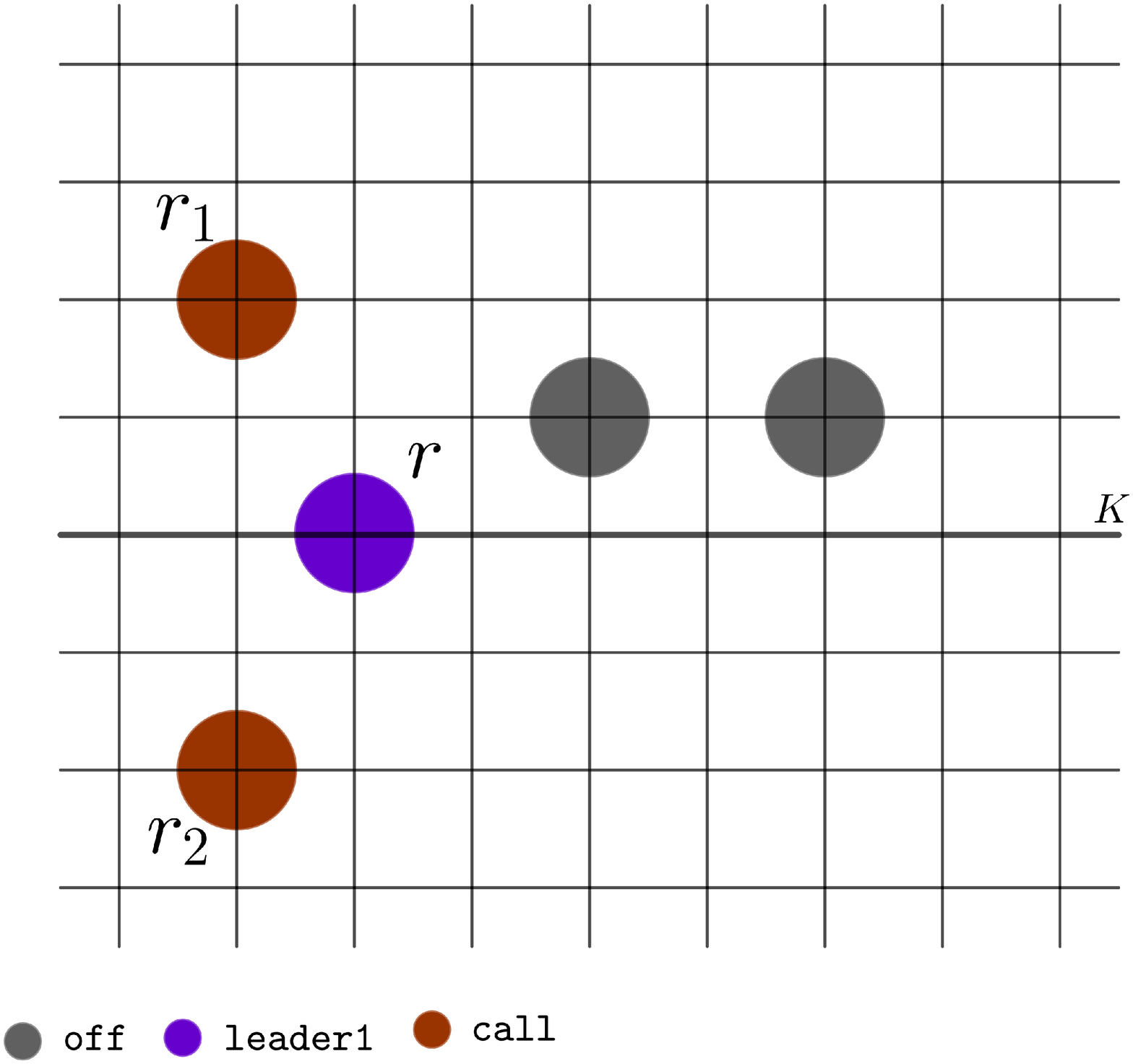}
     \caption{$r_1$ and $r_2$ with colour \texttt{call} see $r$ with colour \texttt{leader1} on $\mathcal{R}_I(r_1)$. No robot with colour \texttt{candidate} in the configuration.}\label{fig:Lemma12case2pic1}
   \end{minipage}
   \hfill
   \begin {minipage}{0.45\textwidth}
    \includegraphics[height=6cm,width=6cm]{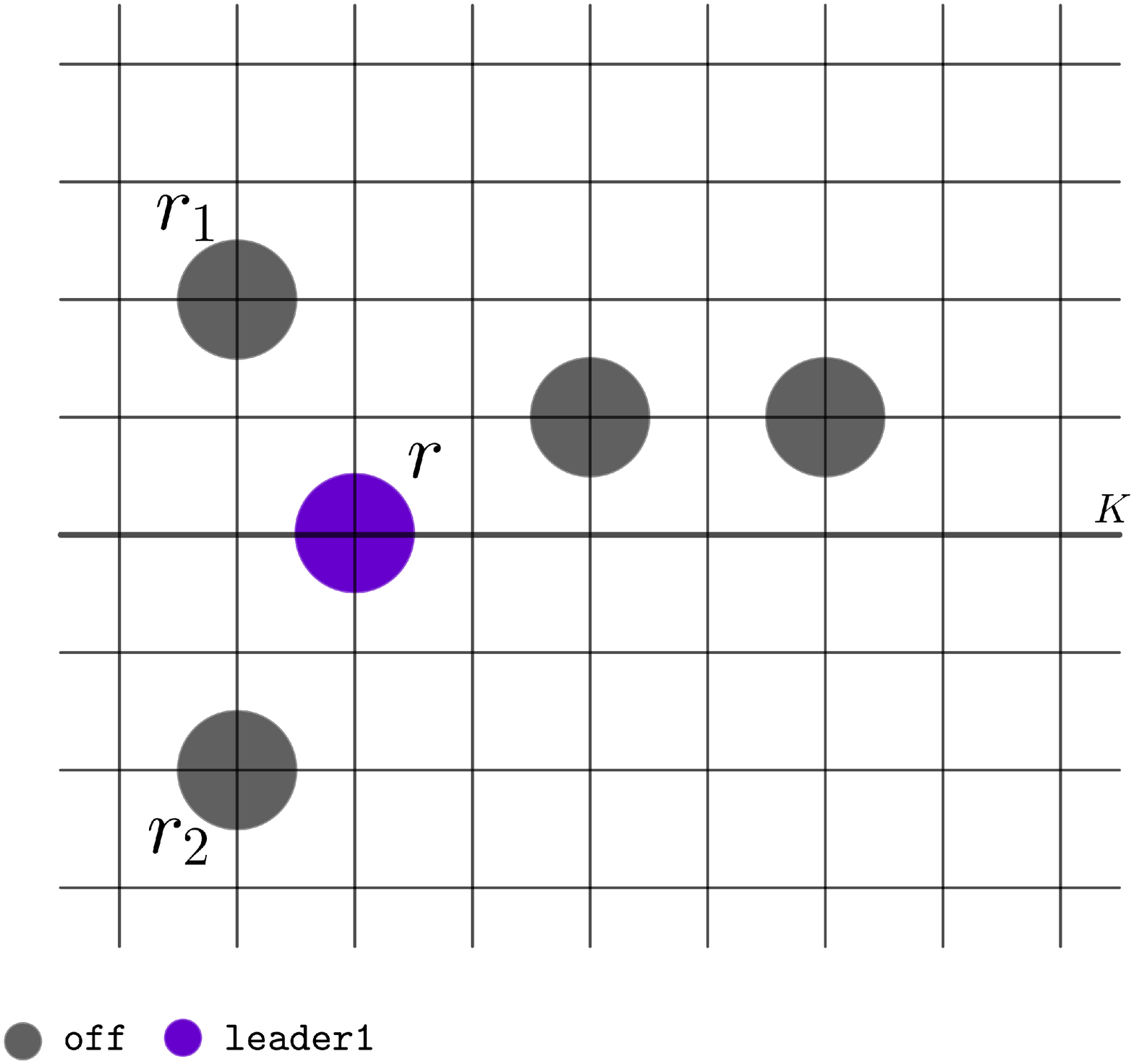}
     \caption{$r_1$ and $r_2$ change their colour to \texttt{off}. No robot with colour \texttt{reached} or \texttt{terminal1} or \texttt{candidate} or \texttt{call} in the configuration and no other robot with colour \texttt{off} changes colour to \texttt{terminal1}  and hence to \texttt{candidate} and hence to \texttt{leader1} or \texttt{call} or directly to \texttt{leader1}. }\label{Fig:Lemma12case2pic2}
   \end{minipage}
\end{figure}

Again this scenario, all robots with colour \texttt{off} are either on $H_R^C(r)$ or on $\mathcal{L}_I(r)$.  Note that all robots with colour \texttt{off} that are on $\mathcal{L}_I(r)$, never change their colours as they either see $r$ on their right immediate vertical line or on the same vertical line or they find their left open half is non-empty. Similarly, the robots with colour \texttt{off} on $H_R^C(r)$ never change their colours again as they find their left open half non-empty, never see exactly two robots with colour \texttt{call} on their left immediate vertical line and never sees a robot with colour \texttt{moving1} on their same vertical line. So, we have proved if a robot has changed its colour to \texttt{leader1} from colour \texttt{off}, no other robot ever changes its colour to \texttt{leader1} again.

So, from all the cases, it is evident that in \textit{Phase 1} any configuration can have at most one robot with colour \texttt{leader1} and since a robot changes its colour to \texttt{leader} from \texttt{leader1} only, there can be at most one robot with colour \texttt{leader} in any configuration during \textit{Phase 1}.
 \end{proof}

\begin{lemma}
If a robot $r$ changes its colour to \texttt{leader1} at a time $T$ and no robot with colour \texttt{terminal1} changes its colour to \texttt{candidate} at $T'$ where $T' \ge T$, then robots on $\mathcal{R}_I(r)$ in $\mathbb{C}(T)$ never move in \textit{Phase 1}.
\end{lemma}
\begin{proof}
Let $r$ changes its colour to \texttt{leader1} from colour \texttt{off} at a time $T$. Observe that in this case, $r$ must have seen two robots say $r_1$ and $r_2$ with colour \texttt{call} at some time $T_1 < T$ and it is on $K \cap \mathcal{R}_I(r_1)$. Note that in this case, all robots on $\mathcal{R}_I(r)$ have colour \texttt{off}, so they do not move. Now upon activation after time $T$, both $r_1$ and $r_2$ change their colours to \texttt{off} and never change their colours again as they see $r$ with colour \texttt{leader1} on $\mathcal{L}_V(r_1)$ or on $\mathcal{R}_I(r_1)$ or other robots on their left open half throughout completion of \textit{Phase 1}. So, robots on $\mathcal{R}_I(r)$(at $T$) never see any robot with colour \texttt{call} and also, they do not see their left open half empty after time $T$ . Thus robots on $\mathcal{R}_I(r)$ at time $T$ never change their colours and never move until completion of \textit{Phase 1}.

Now if $r$ changes its colour from \texttt{candidate} to \texttt{leader1} at time $T$ and no robot with colour \texttt{terminal1} changes its colour to \texttt{candidate} at some time $T'$ where $T' \ge T$, then all robots which are on $\mathcal{R}_I(r)$ in $\mathbb{C}(T)$ can have colour either \texttt{moving1} or \texttt{off} or, \texttt{terminal1}. Note that the robots with colour \texttt{moving1} will not move as it does not see any robot with colour \texttt{call} on its left immediate vertical line or, robot with colour \texttt{reached} on its same vertical line and on its left immediate vertical line. Similarly robots with colour \texttt{off} on $\mathcal{R}_I(r)$ does not change its colour if it wakes after $T$ as it finds out that it has its left open half non-empty and there is no robot with colour \texttt{call} on its left immediate vertical line. Note that a robot with colour \texttt{terminal1} may change its colour to \texttt{off} but after that this robot with colour \texttt{off} will not move by similar argument above. So, no robot on $\mathcal{R}_I(r)$ ever moves after time $T$ until \textit{Phase 1} is complete.
\end{proof}

\begin{lemma}
\label{last_ph1}
If a robot $r$ changes its colour to \texttt{leader1} from \texttt{candidate} at some time $T$ and another robot $r'$ changes its colour to \texttt{candidate} at a time $T'$ where $T' \ge T$, then no collision occurs even if both $r$ and $r'$ move.
\end{lemma}
\begin{proof}
Let $r$ changes its colour to \texttt{leader1} from \texttt{candidate} at a time $T$ and $r'$ changes its colour to \texttt{candidate} from colour \texttt{terminal1} at a time $T' \ge T$. Note that at time $T$, $r$ must be singleton on $\mathcal{L}_V(r)$. This implies there is a time $T_1 < T$ when $r$ had colour \texttt{off} and was terminal on $\mathcal{L}_1$ in $\mathbb{C}(T_1)$. Now when $r$ wakes at a time say $T_2$, where $ T_1 \le T_2 < T$, it changes its colour to \texttt{terminal1}  and there exist a time $T_3 > T_2 \ge T_1$ and $T_3 < T$ such that $r$ changes its colour to \texttt{candidate} and moves left and become singleton on $\mathcal{L}_V(r)$. We claim that $r'$ changes its colour to \texttt{candidate} only if $r'$ was on $\mathcal{L}_V(r)$ in $\mathbb{C}(T_2)$ and it was also terminal on $\mathcal{L}_V(r)$ (i.e $\mathcal{L}_V(r')$) as otherwise $r'$ can not see its left open half empty. Now, let $r'$ wakes before time $T_2$ and decide to changes its colour to \texttt{candidate} but it changes its colour at a time $T' \ge T$ and has a pending move. Then observe that now $r$ and $r'$ are on two different vertical lines and $r'$ has a pending move to the left. So, if there are other vertical lines in between $\mathcal{L}_V(r)$ and $\mathcal{L}_V(r')$, then even if both of them move, no collision occurs as $r$ can only move either vertically or on left and they are on different horizontal line. So, let us consider $r'$ is on $l_{next}(r)$. Then $r$ can not move vertically as $l_{next}(r)$ is non-empty. Hence, both $r$ and $r'$ move left and no collision occurs as $r$ and $r'$ are on different horizontal lines.
\end{proof}

Now, from the above lemmas and the discussions, we can conclude the following theorem.
 
 \begin{theorem}
For any initial configuration $\mathbb{C}(0)$, there exists a $T>0$ such that $\mathbb{C}(T)$ has exactly one robot with colour \texttt{leader} and it's left closed half and one of upper and bottom closed half have no other robots.
\end{theorem}
\subsection{Phase 2}
After completion of \textit{Phase 1}, the configuration has exactly one robot $r_0$ with colour \texttt{leader} such that $r_0$ is singleton on $H_L^C(r_0)$ and also singleton on $\mathcal{L}_H(r_0)$ and there is no other robot on either below or above $\mathcal{L}_H(r_0)$. Note that in this configuration, all the robots who can see $r_0$ can agree on a global coordinate. Let $r_1$ be a robot which can see $r_0$. Then it assumes the position of $r_0$ as the coordinate $(0,-1)$. Now since all robots agree on the direction and orientation of the $x$-axis (i.e the horizontal lines), $r_1$ can think of the horizontal line let's say $\mathcal{H}$ which is just above $r_0$ as the $x$-axis where right half of the line of $\mathcal{L}_V(r_0)$ correspond to the positive direction of $x-$axis. Now $r_1$ agrees on the the vertical line $\mathcal{L}_V(r_0)$ as $y$-axis. Note that $r_1$ can also know the orientation of $y$-axis by assuming its own $y-coordinate$  to be  greater or equals to the $y-coordinate$ of $r_0$ (i.e if $H_U^C(r_0)$ has robots, then  the direction of $\mathcal{L}_V(r_0)$ from $r_0$ towards $H_U^C(r_0) \cap \mathcal{L}_V(r_0)$ is the direction of positive $y-axis$ and similar for the case if $H_B^C(r_0)$ has robots) (Figure \ref{fig:GlobalCoordinateSet}). In this phase, robots first form a line and then from that line move to their corresponding target positions which are embedded in the grid assuming the global coordinate which has been agreed upon by robots after seeing $r_0$. Thus the pattern formation is done. We provided a detailed description of the algorithm for \textit{Phase 2}.

\begin{algorithm}[H]
\footnotesize

    
{

    $r \leftarrow$ myself
    
    $r_0 \leftarrow$ the robot with light \texttt{leader}
    
    \uIf{$r.light =$ \texttt{moving1} or \texttt{candidate} or \texttt{terminal1}}{
    
      \If{$($$r_0 \in {H}_B^O(r)$$)$  and $($$r$ is leftmost on $\mathcal{L}_H(r)$$)$ and $($there is no robot in ${H}_{B}^{O}(r) \cap {H}_{U}^{O}(r_0))$}
                {\If{there are $i$ robots on $\mathcal{L}_H(r_0)$ other than $r_0$ at $(1,-1), \ldots, (i,-1)$}{$r.light \leftarrow$ \texttt{off}\\move to an empty grid point towards $(i+1,-1)$ by \textsc{GotoLine}}}
      
                                       }
    
    \uElseIf{$r.light =$ \texttt{off}}{
    
      \uIf{$($$r_0 \in {H}_B^O(r)$$)$  and $($$r$ is leftmost on $\mathcal{L}_H(r)$$)$ and $($there is no robot in ${H}_{B}^{O}(r) \cap {H}_{U}^{O}(r_0))$ \label{code ss2: 0}}{
    
	  \uIf{there are no robots on $\mathcal{L}_H(r_0)$ other than $r_0$ \label{code ss2: 1}}{
	  
	    \uIf{there is a robot with light \texttt{done}\label{code ss2: 2}}{
	    
	    \uIf{$r$ is at $t_{n-2}$}{$r.light \leftarrow$ \texttt{done}}
	    \Else{move to an empty grid point towards $t_{n-2}$ by \textsc{GotoTarget}}
	    
	    }
	    \Else{move to $(1,-1)$ by \textsc{GotoLine}}
	  
	  }
	  \uElseIf{there are $i$ robots on $\mathcal{L}_H(r_0)$ other than $r_0$ at $(1,-1), \ldots, (i,-1)$\label{code ss2: 3}}{move to a empty grid point towards $(i+1,-1)$ by \textsc{GotoLine}}
	  \ElseIf{there are $i$ robots on $\mathcal{L}_H(r_0)$ other than $r_0$ at $(n-i,-1), \ldots, (n-1,-1)$\label{code ss2: 6}}{
	  
	    \uIf{$r$ is at $t_{n-i-2}$}{$r.light \leftarrow$ \texttt{done}}
	    \Else{move to an empty grid point towards $t_{n-i-2}$ by \textsc{GotoTarget}}
	  
	  }
    
      }
      
      \ElseIf{$r_0 \in \mathcal{L}_H(r)$ and ${H}_U^O(r)$ has no robots with light \texttt{off}\label{code ss2: 4}}{
      
	\If{$r$ is at $(i,-1)$}{move to $(i,0)$}
      
      }

    }

    \ElseIf{$r.light =$ \texttt{leader}}{
    
      \If{all robots, which are visible to $r$, have light \texttt{done} \label{code ss2: 5}}{
    
    \uIf{$r$ is at $t_{n-1}$}{$r.light \leftarrow$ \texttt{done}}
	    \Else{move to an empty grid point towards $t_{n-1}$ by \textsc{LeaderMove}} }
    
      }
    }

\caption{\textsc{ApfFatGrid}: Phase 2}
    \label{algo_ph3} 
\end{algorithm}

\begin{figure}[ht]
    \centering
    \includegraphics[height=6cm,width=7cm]{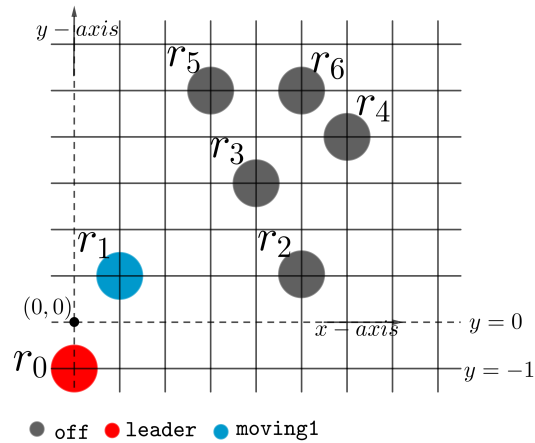}
    \caption{$r_0$ with colour \texttt{leader} is at $(0,-1)$. Any robot that sees $r_0$ can agree on a global coordinate system as shown in this diagram.}
    \label{fig:GlobalCoordinateSet}
\end{figure}

\subsubsection{Line Formation} In the beginning of \textit{Phase 2}, if a robot $r$ sees $r_0$ with colour \texttt{leader1}, it agrees on a global coordinate as mentioned above and find that $r_0$ is on $H_B^O(r)$. Now, if $r$ sees there is no robot in $H_B^O(r) \cap H_U^O(r_0)$ (i.e there is no horizontal line containing any robot between $\mathcal{L}_H(r_0)$ and $\mathcal{L}_H(r)$) and $r$ is leftmost on $\mathcal{L}_H(r)$ and also if it finds there are $i$ other robots on $ (1,-1),(2,-1), \dots ,(i,-1)$ and no robot with colour \texttt{done}, then it changes its colour to \texttt{off} (if $r$ has colour \texttt{off}, it would not change the colour) and moves to $(i+1,-1)$ by a method \textsc{Gotoline()}. The method \textsc{Gotoline()} is described as follows. In this method, if $r$ is above the horizontal line where $y=0$, then it moves vertically downwards until it reaches the line where $y=0$. Note that no collision occurs during this vertical movement as there are no robots between $\mathcal{L}_H(r)$ and $\mathcal{L}_H(r_0)$ and no other robot will move until it reaches at $(i+1,-1)$. This is because other robots even if gets activated before $r$ reaches $(i+1,-1)$, sees $r$ between its horizontal line and $\mathcal{L}_H(r_0)$. Now, when $r$ is at the horizontal line where $y = 0$, it moves horizontally to the position $(i+1,0)$. Note that the horizontal line $y=0$ may not be initially empty at the beginning of \textit{Phase 2}. Also, initially all robots on line $y=0$ have $x-coordinate > 0$. Now let at some time $T$, the robot $r$ which is leftmost on the  line  $y=0$ sees robots on $(1,-1),(2,-1), \dots (i,-1)$. Now if there is no other robot except $r$ on line $y=0$ and $r$ is not on $(i+1,0)$, then $r$ can simply move horizontally without collision to reach $(i+1,0)$. Note that during this movement, no other robot from above moves as they see $r$ in $H_B^O(r) \cap H_U^O(r_0)$. Now if there are other robots on the line $y=0$ other than $r$, then this implies $r$ has $x-coordinate > i$. This is because the $i$ robots say $r_1,r_2,\dots ,r_i$ must have reached their current position  at $(1,-1),(2,-1),\dots (i,-1)$ from the line $y=0$ and all robots on line $y=0$ have $x-coordinate >0$. Now since $r$ has $x-coordinate > i$, $r$ will move to its left until it reaches $(i+1,0)$. Note that during this movement, no collision occurs as $r$ is leftmost robot on $\mathcal{L}_H(r)$ and no other robot on $\mathcal{L}_H(r)$ moves as they are not leftmost on $\mathcal{L}_H(r)$. Next after $r$ reaches the position $(i+1,0)$, it moves vertically downward to $(i+1,-1)$ (Figure \ref{fig:LineFormation}). So after a finite time, all robots will reach on the line $y=-1$ in such a way that there is no empty grid point between two robots on the line $y=-1$ (Figure \ref{Fig:LineFormed}). From the above discussion, we have the following Lemmas \ref{flemmap2col1} and \ref{flemmap2col11}.

\begin{figure}[!htb]\centering
   \begin{minipage}{0.45\textwidth}
    \includegraphics[height=6cm,width=6cm]{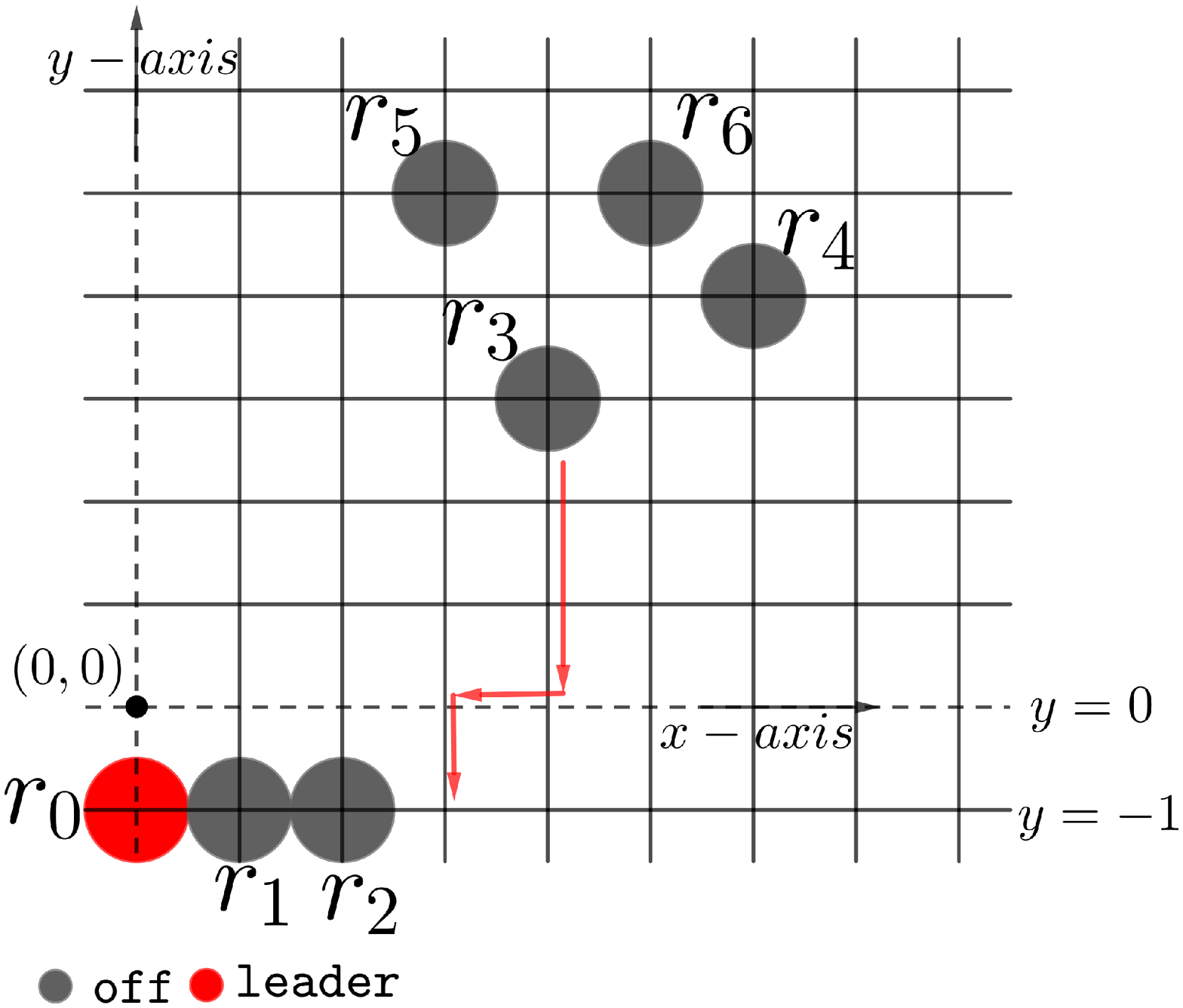}
     \caption{Movement of $r_3$ to $\mathcal{L}_H(r_0)$ by executing the method \textsc{GoToLine()}.}\label{fig:LineFormation}
   \end{minipage}
   \hfill
   \begin {minipage}{0.45\textwidth}
    \includegraphics[height=6cm,width=6cm]{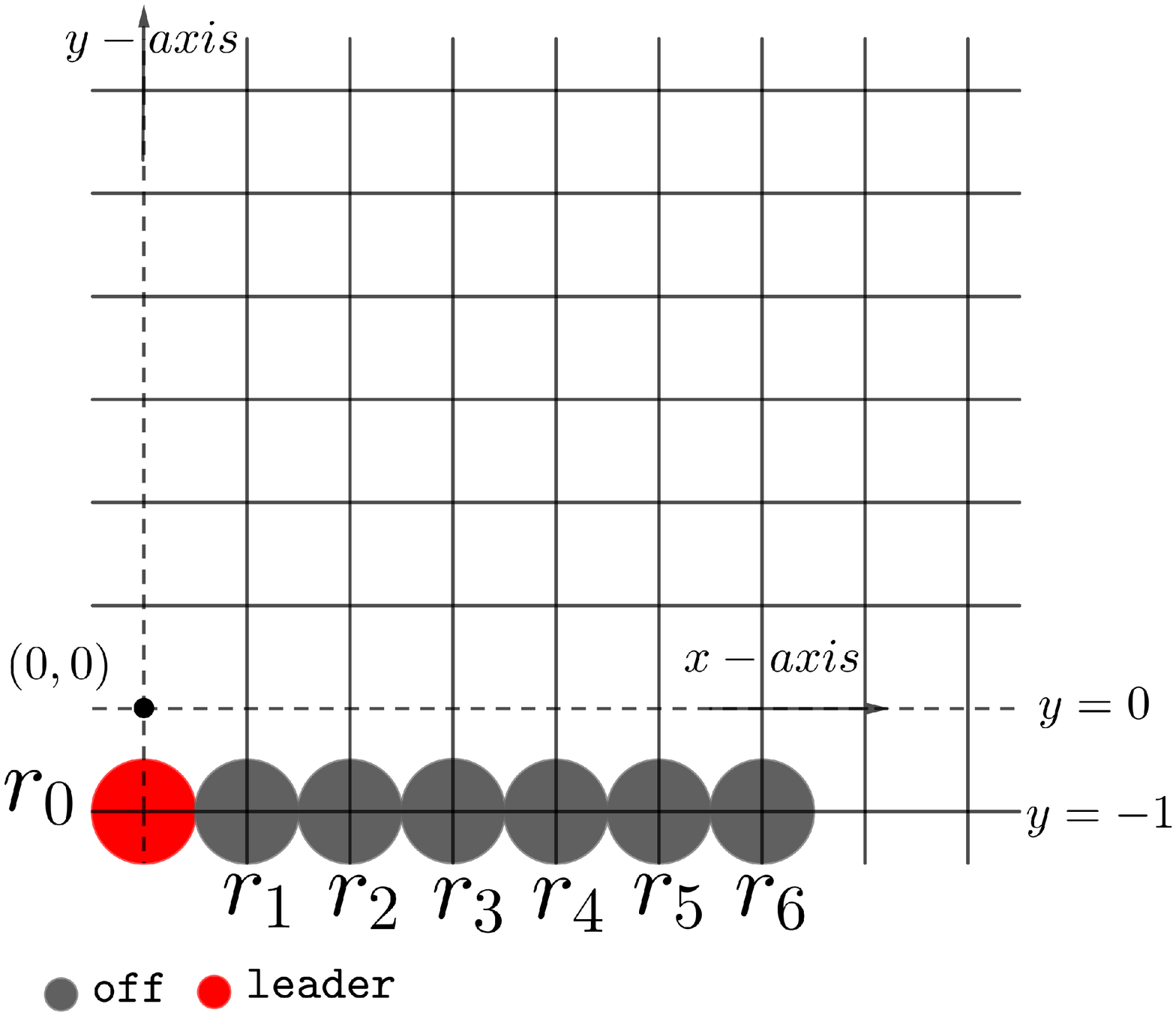}
     \caption{All robots formed a line on $\mathcal{L}_V(r_0)$. There is no empty grid point between any two robots on the line.}\label{Fig:LineFormed}
   \end{minipage}
\end{figure}

\begin{lemma}
\label{flemmap2col1}
During movement of a robot $r$, that is executing the method \textsc{Gotoline()} in \textit{Phase 2}, $r$ does not collide with any other robot in the configuration.
\end{lemma}

\begin{lemma}
\label{flemmap2col11}
There exists a $T>0$ such that $\mathbb{C}(T)$ has one robot with light \texttt{leader} and all other robots with light \texttt{off} in same horizontal line.
\end{lemma}

\begin{lemma}
The leftmost robot $r$ of a horizontal line can always see all the robots on the horizontal line $y = -1$ if there is no robot in $H_B^O(r) \cap H_U^O(r_0)$, where $r_0$ is the robot with colour \texttt{leader} on the line $y=-1$.
\end{lemma}
\begin{proof}
Let $r$ be the leftmost robot on a horizontal line $y = s$ where $s>0$. Then it is obvious that $r$ will see all robots on $y = -1$. 
Now if $r$ is on $y=0$ and is singleton, then also it is obvious that $r$ can see all robots on $y = -1$. Now if $r$ is not singleton on $y=0$. Then from the discussions above, it is clear that $x-coordinate$ of $r > x-coordinate$ of rightmost robot on $y=-1$. So, $r$ can see all robots on $y =-1$.
\end{proof}

\subsubsection{Target Pattern Formation}
After the line is formed, all the robots except $r_0$ with colour \texttt{leader} on $(0,-1)$ have colour \texttt{off} and they are all placed on the horizontal line $y = -1$ in such a way that $r_0$ is the leftmost robot on $y = -1$ and any two robots do not have any unoccupied grid points between them. Now, the target pattern is embedded on the grid (based on the global coordinate system described above) such that $x-coordinate, y-coordinate$  of any target position say, $t_i$ is greater than 0. Also, the target position of robot with colour \texttt{leader}, $t_{n-1}$ is on line $y=1$. And for any two other robots, let there be two target positions $t_i$ and $t_j$. If both $t_i$ and $t_j$ are on the same horizontal line and $i < j$, then $t_i$ is at the right of $t_j$. Also, if $t_i$ and $t_j$ are on two different horizontal lines and $i < j$, then $t_i$ is on the above horizontal line (Figure \ref{fig:targetEmbedding}). 
\begin{figure}[ht]
    \centering
    \includegraphics[height=6cm,width=7cm]{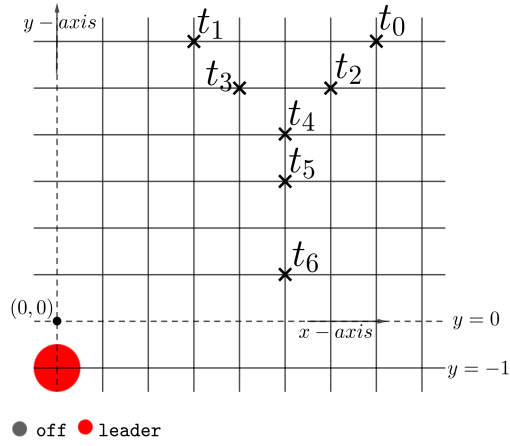}
    \caption{The target pattern embedded in the coordinate system.}
    \label{fig:targetEmbedding}
\end{figure}

Note that a robot say $r$ who can see the robot $r_0$ with colour \texttt{leader} can agree on the global coordinate system as described earlier. So, $r$ can agree on its target position which is embedded target positions on the grid. Now a robot $r$, who can see $r_0$ and has coordinate $(s, -1)$ on $\mathcal{L}_H(r)$ and sees $H_U^O(r)$ has no robot with colour \texttt{off}, moves to position $(s,0)$ (Figure \ref{fig:targetform1}). Now, $r$ can see all robots on $\mathcal{L}_H(r_0)$ as $r$ is singleton on $y=0$. Now, if $r$ finds out that there are $i$ robots on $\mathcal{L}_H(r_0)$ other than $r_0$ at the positions $(n-i , -1), (n-i-1, -1), \dots (n-1, -1)$, then $r$ moves to $t_{n-i-2}$ by a method \textsc{GotoTarget()} (Figure \ref{Fig:targetform2}) and changes its colour to \texttt{done}. The method \textsc{GotoTarget()} is described as follows. When executing the method \textsc{GoToTarget()}, a robot $r$ first moves vertically to the horizontal line that is just below the horizontal line of its target location say $t_r$. Note that during this movement, no other robot on $\mathcal{L}_H(r_0)$ moves even if they see $r_0$ as they see $r$ with colour \texttt{off} on their upper open half. Now, let coordinate of $t_r$ be $(x_{t_r}, y_{t_r})$. Note that after the vertical movement, $r$ is now singleton on the line $y = y_{t_r}-1$. Now, if $r$ is not already on the position $(x_{t_r}, y_{t_r}-1)$, it moves horizontally to the position $(x_{t_r}, y_{t_r}-1)$. Note that since $r$ is singleton on $y = y_{t_r}-1$ and no robot from $\mathcal{L}_H(r_0)$ has started its vertical movement (this is because they still see $r$ with colour \texttt{off} on their upper open half), no collision occurs during this movement by $r$. Now when $r$ reaches the position $(x_{t_r}, y_{t_r}-1)$, it moves above once and reaches the designated target position $t_r$ of $r$. Observe that the last robot say $r_{n-1}$ on $(n-1, -1)$ (i.e on $\mathcal{L}_H(r_0)$) sees no other robot except $r_0$ on $\mathcal{L}_H(r_0)$ after it starts moving vertically above. Now, the problem is if it can distinguish whether $r_{n-1}$ is meant to execute \textsc{GoToLine()} or \textsc{GoToTarget()} when it is above the line $\mathcal{L}_H(r_0)$ . Note that $r_{n-1}$ will see at least one robot having colour \texttt{done} above it or on the same line while it is meant to execute \textsc{GoToTarget()} as $n > 2 \implies n-1 > 1$ which implies $\mathcal{L}_H(r_0)$ had at least one other robot $r_{n-2}$ between $r_0$ and $r_{n-1}$ which already executed \textsc{GoToTarget()} and changed its colour to \texttt{done} before $r_{n-1}$ started executing \textsc{GoToTarget()}. So in this case, $r_{n-1}$ sees at least one robot with colour \texttt{done}  and moves to its designated target location $t_{r_{n-1}}= t_{n-2}$. So, we can conclude that after a finite time, all robots except the robot $r_0$ with colour \texttt{leader} move to their designated target locations embedded on the grid as described earlier. 

\begin{figure}[!htb]\centering
   \begin{minipage}{0.45\textwidth}
    \includegraphics[height=5.5cm,width=6cm]{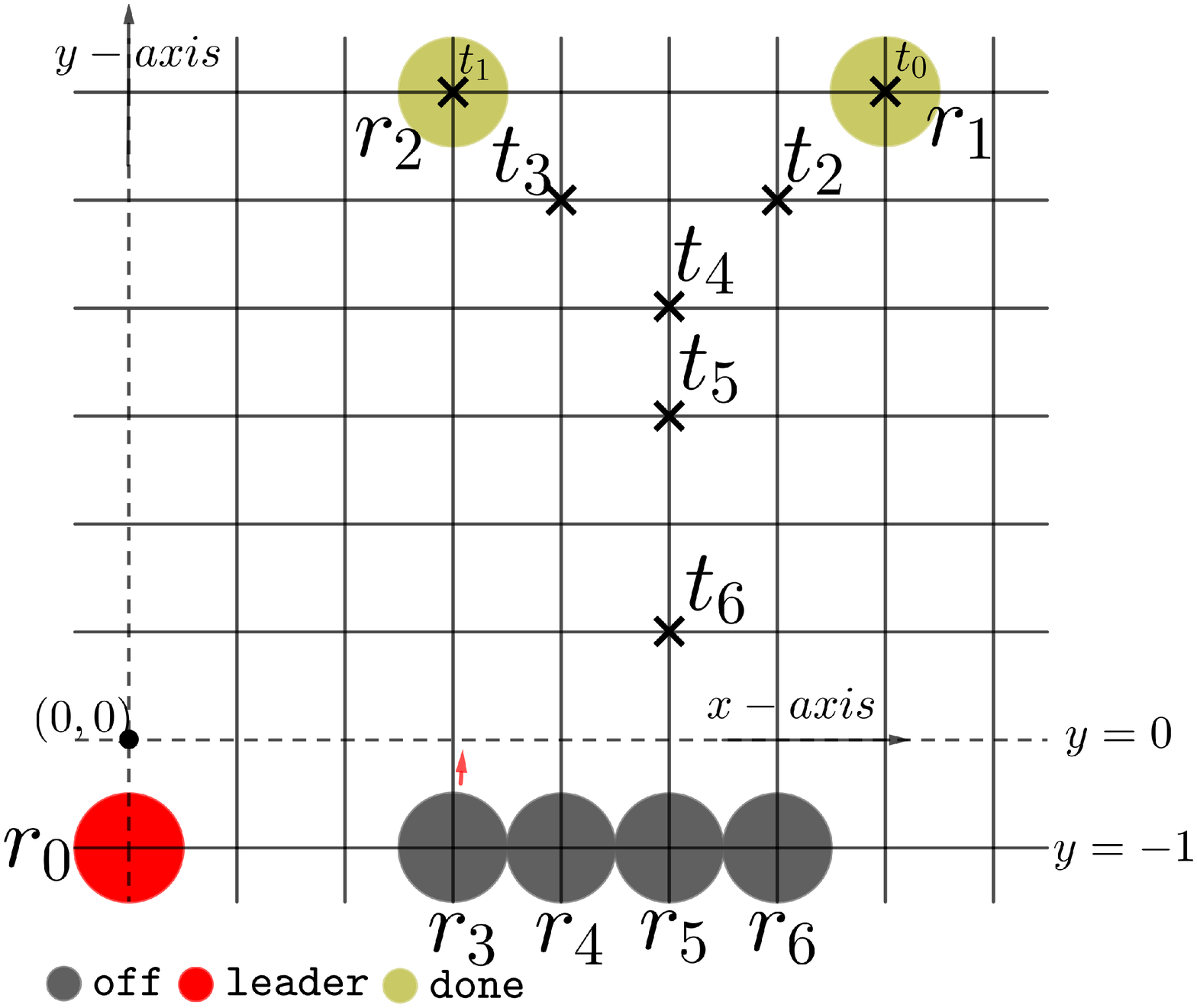}
     \caption{$r_1$ and $r_2$ already moved to their target position and changed their colour to \texttt{done}. $r_3$ sees $r_0$ with colour \texttt{leader} and moves to line $y=0$ by moving vertically.}\label{fig:targetform1}
   \end{minipage}
   \hfill
   \begin {minipage}{0.45\textwidth}
    \includegraphics[height=5.5cm,width=6cm]{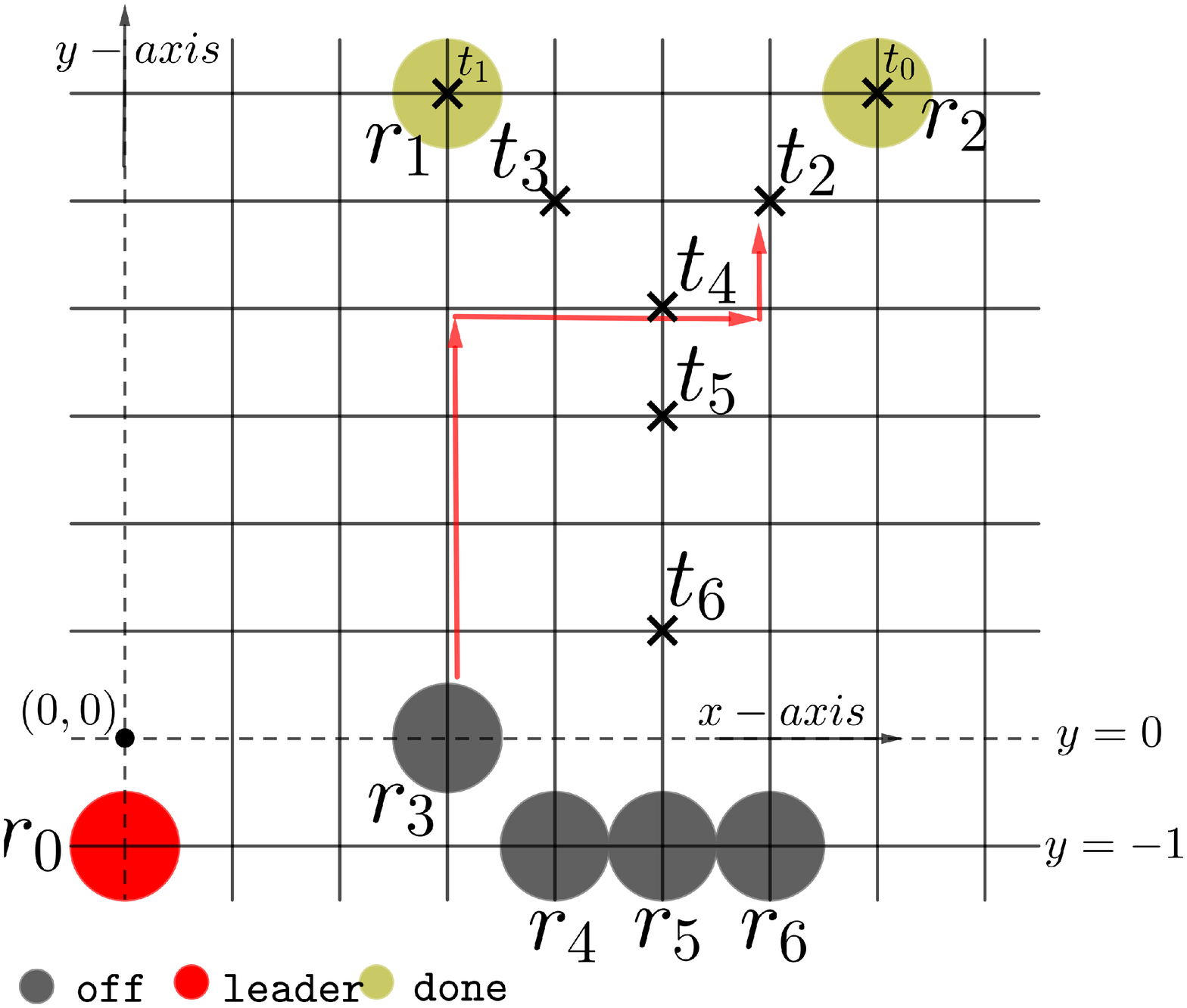}
     \caption{From $(3,0)$, $r_3$ can see 3 robots on $(4,-1), (5,-1), (6,-1)$ and moves to $t_{2}$ by executing the method \textsc{GoToTarget()}. }\label{Fig:targetform2}
   \end{minipage}
\end{figure}
\begin{figure}[!htb]\centering
   \begin{minipage}{0.45\textwidth}
    \includegraphics[height=5cm,width=6cm]{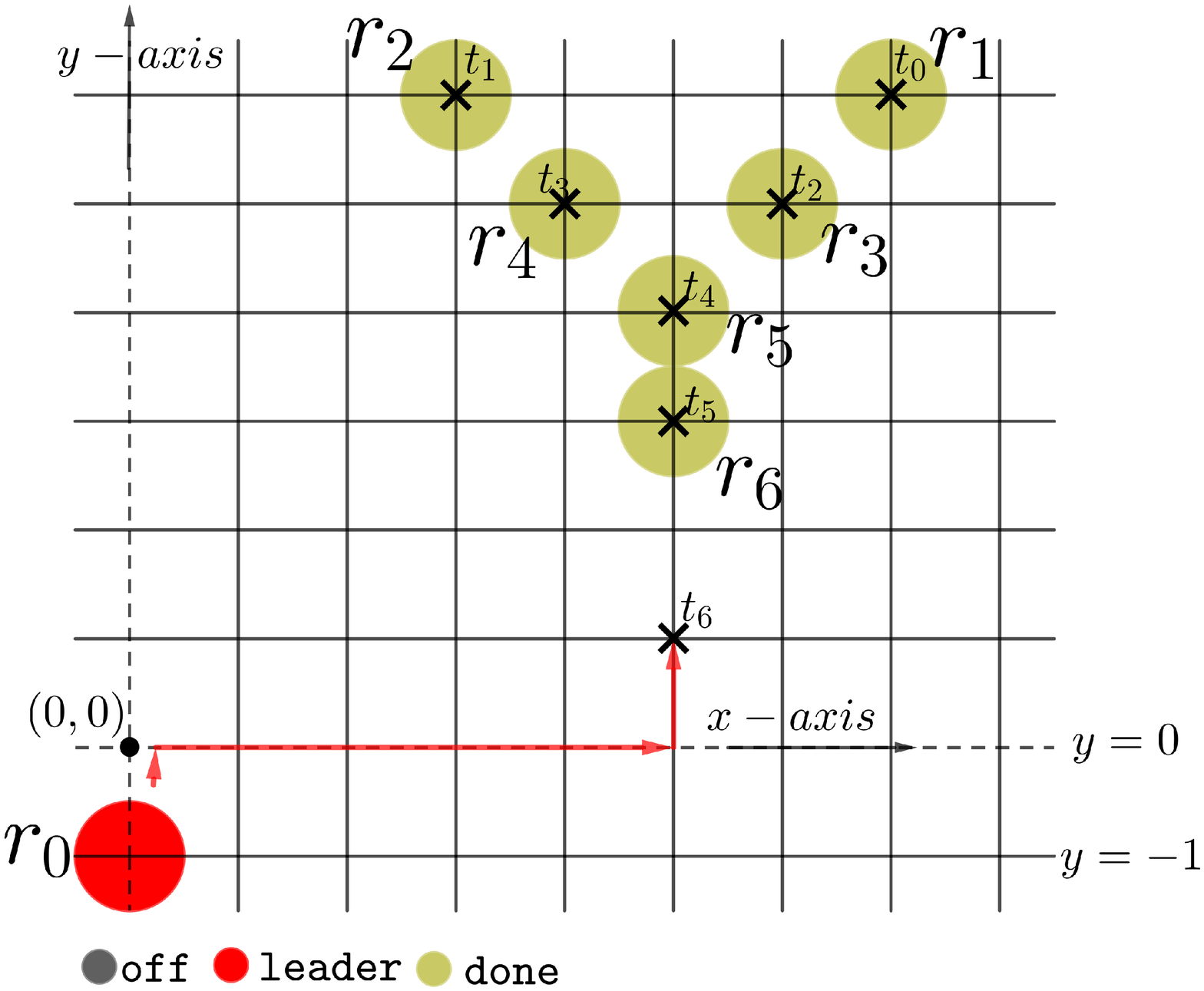}
     \caption{When all other robots except $r_0$ with colour \texttt{leader} reach their corresponding target positions, $r_0$ moves to $t_6$ executing the method \textsc{LeaderMove()}.}\label{fig:targetform3}
   \end{minipage}
   \hfill
   \begin {minipage}{0.45\textwidth}
    \includegraphics[height=5cm,width=6cm]{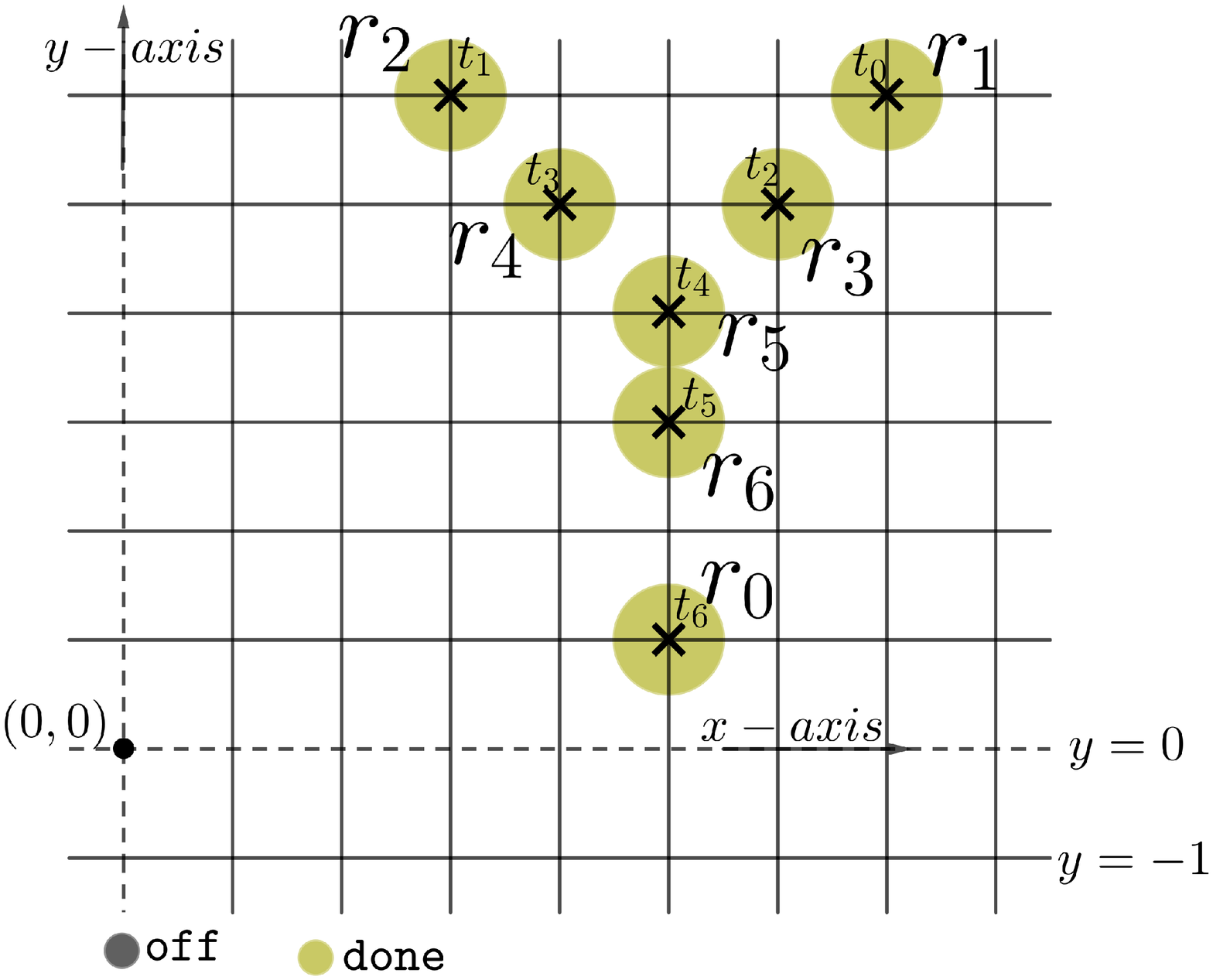}
     \caption{After $r_0$ reaches $t_6$, it changes its colour \texttt{done} and the target pattern has been formed.}\label{Fig:targetformed}
   \end{minipage}
\end{figure}

From the above discussion, we can conclude the following lemma. 
\begin{lemma}
\label{flemmap2col2}
During the execution of the method \textsc{GoToTarget()}, a robot $r$ never collides with another robot in the configuration.
\end{lemma}

Now the robot $r_0$ with colour \texttt{leader} sees that all the visible robots  have colour \texttt{done}. So, it now moves to its designated target location $t_{n-1}$ by a method \textsc{LeaderMove()}. The method \textsc{LeaderMove()} is described as follows. In this method, $r_0$ first moves to $(0,0)$. Now, let the lowest horizontal line be $\mathcal{H}_{last}$ having a robot with colour \texttt{done}. Now, note that $r_0$ can always see the leftmost robot $r_{n-1}$ on the horizontal line $\mathcal{H}_{last}$. So, $r_0$ can always know its own position on the global coordinate as it knows the target position of $r_{n-1}$ from the input even if it is not at $(0,-1)$. Now the target was embedded in such a way that the target position $t_{n-1}$ of $r_0$ is on line $y=1$. Let $(x_{t_0},1)$ be the target position of $r_0$. Now from $(0,0)$, $r_0$ moves horizontally to the location $(x_{t_0},0)$ and then moves vertically once to $t_{n-1} = (x_{t_0},1)$ and changes the colour to \texttt{done}. Note that below $y=1$, there is no other robot while $r_0$ starts moving (Figure \ref{fig:targetform3}). So, we can conclude the following lemma.

\begin{lemma}
\label{flemmap2col3}
While executing the method \textsc{LeaderMove()}, a robot $r$ never collides with other robots in the configuration.
\end{lemma}


So, from Lemmas \ref{flemmap2col1}, \ref{flemmap2col2} and \ref{flemmap2col3}, we can directly conclude the following result.
\begin{lemma}
During movement of robots in \textit{Phase 2}, no collision occurs.
\end{lemma}
Now from the above lemmas and the discussions, we can now finally conclude the following theorem.
 
 \begin{theorem}
There exists a $T>0$ such that $\mathbb{C}(T)$ is a final configuration similar to the given pattern and has all robots with light \texttt{done} (Figure \ref{Fig:targetformed}).
\end{theorem}

\section{Conclusion}
The problem of arbitrary pattern formation ($\mathcal{APF}$) is a widely studied area of research in the field of swarm robotics. It has been studied under various assumptions on plane and discrete domain (eg. infinite regular tessellation grid). With obstructed visibility model, this problem has been considered on plane and infinite grid using luminous opaque robots. But using fat robots (i.e. robots with certain dimensions), it is only done in plane. In this paper, we have taken care of this. We have shown that with a swarm of luminous opaque fat robots having one-axis agreement on an infinite grid, any arbitrary pattern can be formed from an initial configuration which is either asymmetric or has at least one robot on the line of symmetry using one light having 9 distinct colours which are less than the number of colours used to form an arbitrary pattern on plane using opaque and fat robots with one-axis agreement. For future courses of research, it would be interesting to see if the same problem can be solved using less number of colours under the same assumptions.

\vspace{1cm}
\textbf{Acknowledgements:} The second author is supported by UGC, the Government of India and the third author is supported by the West Bengal State government Fellowship Scheme. We thank the anonymous reviewers for their valuable comments which helped us to improve the quality and presentation of this paper.

\bibliographystyle{tfnlm}
\bibliography{interactnlmsample}

\end{document}